\newif\ifDRAFT 
\DRAFTfalse

\documentclass[11pt]{article}
\usepackage[margin=1in]{geometry}
\usepackage{graphicx} %
\usepackage[framemethod=TikZ]{mdframed}
\usepackage{amsfonts, amsmath, amssymb, amsthm}
\allowdisplaybreaks
\usepackage{mathrsfs}  
\usepackage{algorithm}
\usepackage[noend]{algpseudocode}
\usepackage{hyperref}
\usepackage{wrapfig}
\usepackage{color}
\usepackage{array}
\usepackage{cleveref}
\definecolor{DarkRed}{rgb}{0.5,0.1,0.1}
\definecolor{DarkBlue}{rgb}{0.1,0.1,0.5}
\definecolor{ForestGreen}{rgb}{0.1333,0.5451,0.1333}
\hypersetup{colorlinks=true, linkcolor=DarkRed,citecolor=ForestGreen}
\usepackage[inkscapeformat=pdf]{svg}
\usepackage{multirow}
\usepackage{mathtools}
\usepackage{enumitem}
\usepackage{thm-restate}
\usepackage{tikz}

\usepackage{tabulary} %
\usepackage{booktabs}
\usepackage{arydshln}
\usepackage[x11names, table]{xcolor}

\ifDRAFT
    \usepackage{lineno}
    \linenumbers
\fi

\theoremstyle{plain}
\newtheorem{theorem}{Theorem}[section]
\newtheorem{lemma}[theorem]{Lemma}

\newtheorem{observation}[theorem]{Observation}
\newtheorem{claim}[theorem]{Claim}

\newtheorem{conjecture}[theorem]{Conjecture}

\theoremstyle{definition}
\newtheorem{definition}[theorem]{Definition}

\newtheorem{remark}[theorem]{Remark}

\crefname{equation}{Eqn.}{Eqns.}

\newcommand{\rsquig}{\rightsquigarrow}

\DeclareSymbolFont{bbold}{U}{bbold}{m}{n}
\DeclareSymbolFontAlphabet{\mathbbold}{bbold}

\newcommand{\IGNORE}[1]{}

\newcommand{\poly}{\operatorname{poly}}

\newcommand{\dist}{\operatorname{dist}}
\newcommand{\cover}{\operatorname{cov}}

\newcommand{\depth}{\operatorname{depth}}

\newcommand{\LastE}{\mathtt{LastE}}

\newcommand{\DP}{\mathtt{DP}}
\newcommand{\CFTSP}{\mathtt{CFT\text{-}SP}}

\newcommand{\DELTA}{k}

\usepackage[normalem]{ulem} %

\mdfdefinestyle{MyFrame}{%
    roundcorner=10pt,
    innertopmargin=10pt,
    innerbottommargin=10pt,
    innerrightmargin=10pt,
    innerleftmargin=10pt,
    backgroundcolor=gray!5!white}

\title{Color Fault-Tolerant Distance Preservers: \\
\~{O}ptimal Size in Conditionally \~{O}ptimal Time}

\ifDRAFT
    \author{Anonymous}
\else
    \author{
    Merav Parter\thanks{Weizmann Institute. Email: \texttt{merav.parter@weizmann.ac.il}. Supported by the European Research Council (ERC) under the European Union’s Horizon 2020 research and
    innovation programme, grant agreement No. 949083.}
    \and
    Asaf Petruschka\thanks{Weizmann Institute. Email: \texttt{asaf.petruschka@weizmann.ac.il}.  Supported by an Azrieli Foundation fellowship.}
    }
\fi

\date{}

\begin{document}

\maketitle

\pagenumbering{roman}
\begin{abstract}
    
    We revisit the problem of \emph{fault-tolerant (FT) distance preservers}, when failure events in the network admit a form of correlation modeled as \emph{color faults}.
    
    FT distance preservers are sparse subgraphs that preserve distances between specified pairs of vertices, even after some edge or vertex failures occur.
    In the classical fault model, any set of at most $\DELTA$ edges or vertices might fail (where $\DELTA \geq 1$ is a given parameter).
    Despite extensive research, the classical model admits significant and tantalizing gaps, both in terms of sparsity bounds and of algorithmic efficiency.
    
    In this work, we study the problem in the recently introduced \emph{color fault-tolerant (CFT)} model:
    the given graph $G=(V,E)$ has arbitrary colors on its edges/vertices where each color appears at most $\DELTA$ times, and is susceptible to \emph{color faults}, where the failure of color $c$ causes all the $c$-colored elements to crash.
    Our main contribution is in the \emph{multi-source} setting, where $G$ has a source-set $S \subseteq V$, and the CFT preserver should preserve $S \times V$ distances under any single color fault.
    We show the following results (where $n = |V|$, $m = |E|$):
    \begin{itemize}
        \item There exists a CFT distance preserver $H$ of $G$ with $\tilde{O}(n^{2 - \frac{1}{\DELTA+1}} \cdot |S|^{\frac{1}{\DELTA+1}} )$ edges.\footnote{The $\tilde{O}(\cdot)$ notation suppresses poly-logarithmic factors in $n$.}
        
        \item The above sparsity bound is worst-case optimal up to polylogarithmic terms.
        
        \item There is a combinatorial randomized algorithm that produces a preserver $H$ whose size meets the above optimal sparsity bound, with running time of $\tilde{O}(m \cdot n^{1 - \frac{1}{\DELTA+1}} \cdot |S|^{\frac{1}{\DELTA+1}})$.
        
        \item The above running time is conditionally optimal: a polynomial improvement 
        would refute the combinatorial Boolean Matrix Multiplication (BMM) conjecture.
        Furthermore, the running time remains optimal even if we only require \emph{mild} sparsification to $m^{1-\epsilon}$ edges.
    \end{itemize}
Our (conditionally) tight algorithm relies on a new approach for compressing fault-tolerant distance information in the presence of nearby faults, using the tool of sparse neighborhood covers. We believe that this technique may have further applications in the study of fault-tolerant graph sparsification.

\end{abstract}
\newpage
\tableofcontents
\newpage
\pagenumbering{arabic}

\section{Introduction}

Let $G=(V,E)$ be an undirected unweighted graph and $P \subseteq V \times V$ be a set of vertex pairs.
A subgraph $H = (V,E')$ of $G$ is called a \emph{$P$-distance preserver} if $\dist_H (u,v) = \dist_G (u,v)$ for every $(u,v) \in P$.
Since their introduction by Coppersmith and Elkin~\cite{CoppersmithE06}, the graph-theoretic and algorithmic aspects of sparse distance preservers and their numerous applications have been extensively studied~\cite{Bodwin19,Bodwin21,Bodwin22,BodwinW21,Pettie09,BodwinW15,ElkinP16,ChangGMW18,CyganGK13,NeimanS24,HoppenworthXX25}.
However, the above definition falls short when it comes to the crucial aspect of \emph{fault tolerance}.
In real-world networks, some nodes (vertices) or links (edges) might occasionally fail, and it is desirable for the distances in the surviving network to still be preserved in the surviving subgraph.
This led researchers to introduce and study the natural notion of fault-tolerant (FT) distance preservers~\cite{BodwinGPW17,ParterP13,Parter17,ParterP18,ParterP20,ChechikLPR09,DinitzK11}.
The following definition refers to failing vertices (the obvious modification to failing edges appears in parentheses):
the subgraph $H$ is called $\DELTA$-FT $P$-distance preserver if $\dist_{H-F} (u,v) = \dist_{G-F} (u,v)$ for every $(u,v) \in P$ and $F \subseteq V$ ($F \subseteq E$) of size $|F| \leq \DELTA$. 

As stated by~\cite{BodwinGPW17}, finding sparse FT preservers turns out to be a very challenging task. 
Indeed, despite extensive research on the graph-theoretical and algorithmic aspects of $\DELTA$-FT distance preservers, tantalizing gaps and questions remain open to date, some resisting progress for over a decade.
This work takes a different approach: we substitute the $\DELTA$-FT model by the recently-introduced \emph{color faults} model~\cite{PetruschkaST26,PetruschkaST24,ParterPST25}, inspired by older notions such as \emph{shared risk resource groups}~\cite{CoudertDPRV07,Kuipers12}, \emph{hedge connectivity}~\cite{GhaffariKP17} and \emph{label cuts}~\cite{ZPT11}.
Intuitively, the goal is accounting for a form of correlation or dependence between failures; these are neglected by the classical $\DELTA$-FT setting allowing \emph{any} set of up to $\DELTA$ vertices/edges to fail.

We assume that $G$ comes with arbitrary colors on its vertices or its edges, where each color appears at most $\DELTA$ times. (Neighboring vertices/edges can have the same color.)
These represent the ``simplest'' form of dependence between individual failures: elements of the same color fail together.
The failure events to be supported are (single) color faults, where a faulty color $c$ leaves us with the graph $G-c$ where all vertices/edges of color $c$ are failing (i.e., deleted from $G$).
We get the following definition for color fault-tolerant (CFT) preservers:

\begin{definition}
    Given graph $G=(V,E)$ whose vertices (edges) are colored such that each color appears at most $\DELTA$ times, and pairs $P \subseteq V \times V$, subgraph $H = (V,E')$ of $G$ is a \emph{CFT $P$-distance preserver} if for every color $c$ and every $(u,v) \in P$, it holds that $\dist_{G-c} (u,v) = \dist_{H-c}(u,v)$.
\end{definition}

We note that the disjointness of color classes is not crucial for us: all of our results trivially extend if each vertex/edge is assigned with some constant number $\ell \geq 1$ of colors, and we focus on $\ell = 1$ just for sake of readability.
However, the assumed bound $\DELTA$ on the size of color classes is unavoidable: if one allows arbitrarily large color classes, then our results (presented next) show that the CFT preservers in consideration cannot provide meaningful sparsification.

\paragraph{Main Results.}

Our main focus is \emph{multi-source} fault-tolerant distance preservers, where there is a set $S \subseteq V$ of sources and $P = S \times V$.
This setting was first considered by Peleg and Parter~\cite{ParterP13}, and has since attracted considerable attention~\cite{Parter15,ParterP15,GhaffariP16,GuptaK17,BodwinGPW17,Parter20,BodwinP23}.
The most fundamental problem is of a \emph{graph-theoretic} nature: determining the (worst-case) size/sparsity of such preservers.
Parter~\cite{Parter15} showed a lower bound of $\Omega(n^{2-\frac{1}{\DELTA+1}} \cdot |S|^{\frac{1}{\DELTA+1}})$ on the number of edges these preservers require in the $\DELTA$-FT model, for every $\DELTA \geq 1$.
However, a matching upper bound is known only for $\DELTA \in \{1,2\}$~\cite{ParterP13,Parter15,GuptaK17}.
When $\DELTA \geq 3$, the state-of-the-art upper bound is $O(n^{2-\frac{1}{2^\DELTA}} \cdot |S|^{\frac{1}{2^\DELTA}})$, given by~\cite{BodwinGPW17,BodwinP23}.
Resolving this gap for $\DELTA$-FT is a major open problem in the field, which poses challenging barriers as demonstrated by~\cite{BodwinP23}.
However, in the color faults model, we are able to fully resolve the sparsity question, up to logarithmic factors:

\begin{theorem}\label{thm:size-upper-bound}
    For any $n$-vertex colored graph $G = (V,E)$ with color classes of size at most $\DELTA$, and any $S \subseteq V$, there exists a CFT $S \times V$  distance preserver with 
    $\tilde{O}(n^{2-\frac{1}{\DELTA+1}} \cdot |S|^{\frac{1}{\DELTA+1}})$ edges.
\end{theorem}

\begin{theorem}[Informal statement of~\Cref{thm:size-lower-bound}]
    The bound in~\Cref{thm:size-upper-bound} is worst-case optimal, up to polylogarithmic factors.
\end{theorem}

Another important aspect of research is \emph{algorithmic}, namely, how efficiently can one compute such sparse preservers.
The current best for $\DELTA$-FT multi-source preservers is the $\tilde{O}(mn)$ time algorithm of~\cite{BodwinGPW17}, producing a preserver with $\tilde{O}(n^{2-\frac{1}{2^\DELTA}} \cdot |S|^{\frac{1}{2^\DELTA}})$ edges as mentioned before;
for $\DELTA = 2$, \cite{Parter15,GuptaK17} give optimal size guarantees of $\tilde{O}(n^{5/3} \cdot |S|^{1/3})$ edges, but with larger running time.
In this work, we settle the time complexity of \emph{combinatorial} algorithms in the color faults model.
The term ``combinatorial algorithms'' is loosely defined, colloquially referring to algorithms that do not use fast matrix multiplication tricks. 
Combinatorial algorithms are of significant interest from both practical and theoretical perspectives, see~\cite{AbboudFKLM24} for a comprehensive discussion.
A fundamental conjecture 
regarding these is on \emph{combinatorial Boolean Matrix Multiplication (BMM)}, asserting that no combinatorial algorithm can polynomially improve over the na\"ive one (see precise statement in~~\Cref{conj:BMM}).
We give a combinatorial algorithm for computing multi-source CFT distance preservers with optimal size bound, whose improvement would refute this conjecture:

\begin{theorem}\label{thm:alg-upper-bound}
    There is a combinatorial randomized algorithm with the following properties.
    Given a colored graph $G = (V,E)$ with $n$ vertices and $m$ edges, where each color appears at most $\DELTA$ times in $G$, and a source set $S \subseteq V$, the algorithm runs in $\tilde{O}(m \cdot n^{1-\frac{1}{\DELTA+1}} \cdot |S|^{\frac{1}{\DELTA+1}})$ time and outputs w.h.p.\footnote{The term w.h.p.\ (with high probability) means with probability $1 - n^{-q}$ for an arbitrarily large constant $q$.}\ a CFT $S \times V$ distance preserver $H$ with $|E(H)| = \tilde{O}(n^{2-\frac{1}{\DELTA+1}} \cdot |S|^{\frac{1}{\DELTA+1}})$ edges.
\end{theorem}

\begin{theorem}[Informal statement of~\Cref{thm:alg-lower-bound}]\label{thm:alg-lower-bound-informal}
    If there is a combinatorial algorithm for computing CFT multi-source distance preservers that runs polynomially faster than in~\Cref{thm:alg-upper-bound} and guarantees mild (polynomial) sparsification, then the combinatorial BMM conjecture is false.
\end{theorem}
\noindent
The design of our (conditionally) tight algorithm introduces a new way to compress fault-tolerant distance information in the face of nearby faults, through \emph{sparse neighborhood covers}~\cite{ABCP98}. We believe this approach has the potential to extend to other problems in fault-tolerant graph sparsification.
We mention that combinatorial lower bounds of a similar flavor have been shown for the closely-related problem of \emph{replacement paths}~\cite{WilliamsWX22,WilliamsW18,ChechikC19,GuptaJM20}.
Specifically, Vassilevska Williams, Woldeghebriel and Xu~\cite{WilliamsWX22} have shown a lower bound for combinatorial data structures for replacement paths under $\DELTA$ edge/vertex faults in directed graphs between a fixed pair of vertices.
Our lower bound in~\Cref{thm:alg-lower-bound-informal} can also be interpreted as a similar data structure lower bound for undirected graphs in the single/multi-source setting (see~\Cref{remark:data-structure} for details).

\paragraph{Other CFT Preservers.}
To complement our main results, in~\Cref{sec:other-CFT} we explore the landscape of other CFT preservers: \emph{single-pair} or \emph{pairwise} distance preservers, \emph{$+2$ spanners} and \emph{reachability}  or \emph{bounded flow} preservers.
These are interesting to view from the perspective of ``utilizing'' the correlation of failures encoded by the color structure.
For example, for single-pair distance preservers in weighted graphs, no meaningful sparsification guarantee is possible already with $\DELTA=2$ edge/vertex faults~\cite{BodwinGPW17}.
However, in the CFT model with color classes of size up to $\DELTA = O(\log n)$, our results imply an upper bound of $\tilde{O}(n^{3/2})$.
On the other hand, for single-source reachability preservers, we prove that the sparsity in the CFT setting does not improve on $\DELTA$-FT.

\paragraph{Concluding Remarks.} 
The size upper bound of~\Cref{thm:size-upper-bound} also holds when $G$ is \emph{directed} (we prove this in~\Cref{sec:multi-source-directed}).
However, our algorithm in~\Cref{thm:alg-upper-bound} does not directly work for directed graphs, and extending it is an interesting open problem.
As the main ``undirected ingredient'' is the neighborhood cover, it is plausible that directed LDD (low diameter decomposition) could be useful for this purpose.
Another challenge is \emph{derandomization}: here, the main barrier is the use of oblivious hitting sets to hit yet-unknown replacement paths, a key randomized ingredient which is recurrent in this context~\cite{BodwinGPW17,ChechikC19,GuptaJM20,ChechikM20,ChechikZ24}.
Finally, while our work settles the complexity of combinatorial algorithms, the study of \emph{non-combinatorial} algorithms relying on fast matrix multiplications is an intriguing future direction.

\section{Preliminaries}

We use the following notations for the vertex (resp., edge) colored graph $G = (V,E)$.
As usual, we denote $n = |V|$ and $m = |E|$.
The color of vertex (resp., edge) $x$ is denoted by $c(x)$.
The set of all $c$-colored vertices (resp., edges) is $V_c$ (resp., $E_c$).
The notation $G-c$ is a shorthand for $G-V_c$ (resp., $G-E_c$).
We assume that each color appears at most $\DELTA$ times in $G$.
For two vertices $u,v \in V$, $\pi(u,v)$ denotes a shortest $u \rsquig v$ path in $G$, and $\pi(u,v \mid c)$ denotes a shortest $u \rsquig v$ path in $G-c$ where $c$ is some color.
The last edge of $\pi(u,v)$, i.e.\ the edge touching $v$, is denoted by $\LastE(u,v)$.
In a similar fashion, $\LastE(u,v \mid c)$ is the last edge of $\pi(u,v \mid c)$.
We also need notations for balls: the $r$-radius ball of a vertex set $U \subseteq V$ in the graph $G$ is $B_G (U,r) = \{v \in V \mid \dist_G(u,v) \leq r \} = \bigcup_{u \in U} B_G (u,r)$.

\section{Algorithm for Multi-Source CFT Distance Preservers}

In this section, we give our main result, providing the algorithm stated in~\Cref{thm:alg-upper-bound} (note that this algorithm also proves the size bound stated in~\Cref{thm:size-upper-bound}).
Throughout, we fix the given sources $S \subseteq V$.
We assume that $\DELTA = O(\log n)$ as otherwise, the size bound of~\Cref{thm:alg-upper-bound} is trivially $O(n^2)$ and we can just output $H = G$.
Also, we assume that $m \geq n^{2-\frac{1}{\DELTA+1}}
|S|^{\frac{1}{\DELTA+1}}$, as otherwise we can again output $H=G$.
We give the algorithm for the case where $G$ has colored vertices; this is a somewhat arbitrary choice, as the proof for colored edges is virtually the same.\footnote{We chose to show the proof with failing vertices rather than edges because in some fault-tolerant settings, vertex faults are considered harder to deal with than edge faults.
However, except for~\Cref{sec:single-pair}, the proofs in this paper are trivially translated from vertex colors to edge colors and vice versa.}
For convenience, one may assume that shortest paths in $G$, and in $G-c$ for any color $c$, are unique (although this is not necessary for us).
This can be achieved by treating the edges of $G$ as having initially weight $1$, then perturbing the weight of each edge by adding small random noise.

\subsection{Framework and Overview}\label{sec:framework}

To build a valid CFT $S \times V$ distance preserver $H$, our algorithm must include in $H$ a shortest $s \rsquig t$ path from $G-c$, for every triplet $(s,t,c) \in S \times V \times C$.
But in fact, it has been commonly observed (e.g.,~\cite{ParterP16,Parter15,GuptaK17,BodwinGPW17}) that including only the \emph{last edge} of every such path suffices.

\begin{observation}[Last edge observation]\label{obs:last-edge}
    Suppose $H$
    includes the last edge of a shortest $s \rsquig t$ path in $G-c$, for every $(s,t,c) \in S \times V \times C$.
    Then $H$ is a CFT $S \times V$ distance preserver of $G$.
\end{observation}
\begin{proof}
    We show that if $d := \dist_{G-c} (s,t)$, then $H-c$ contains an $s \rsquig t$ path of length $d$, by induction on $d$.
    The base case $d = 0$ is trivial.
    When $d > 0$, the last edge $(t',t)$ of a shortest $s \rsquig t$ path in $G-c$ is also in $H-c$ by assumption.
    As $\dist_{G-c} (s,t') = d-1$, by induction hypothesis, $H-c$ contains an $s \rsquig t'$ path of length $d-1$. 
    Concatenating $(t',t)$ gives the desired $s \rsquig t$ path.
\end{proof}
\noindent
Thus, the algorithm's goal is to include $\LastE(s,t\mid c)$ in $H$ for every triplet $(s,t,c) \in S \times V \times C$.

We define decreasing \emph{distance thresholds} by
$d_i := (n/|S|)^{1-\frac{i}{\DELTA+1}}$ for $i=1,2,\dots,\DELTA$.
Additionally, we define two ``trivial thresholds'' 
by
$d_0 := n$ and $d_{\DELTA+1} := 1$.
We use the following terminology:
    Two vertices $u,v \in V$ are said to be \emph{$i$-near} if $\dist_G (u,v) \leq 2d_i$, and \emph{$i$-far} otherwise.
The factor $2$ 
is due to technical reasons, which will become apparent shortly (in the proof of~\Cref{lem:level-i}).

Next, for each nontrivial threshold $d_i$, $i = 1,2 \dots, \DELTA$, we construct a corresponding \emph{hitting set} $A_i \subseteq V$ of size
$|A_i| = O(\frac{n}{d_i} \log n) = O(n^{\frac{i}{\DELTA+1}} \cdot |S|^{1 - \frac{i}{\DELTA+1}} \log n)$,
with the following property: 
\begin{enumerate}[label=\textbf{(P)}]
    \item For every $s \in S$, $t \in V$, and color $c$, if $P$ is a contiguous subpath of $\pi(s,t \mid c)$ of length $|P| \geq d_i$, then $P$ contains some vertex from $A_i$.
    \label{prop:hitting-set}
\end{enumerate}
This is achieved by taking $A_i$ to be a random vertex subset of size $\Theta(\frac{n}{d_i} \cdot \log n)$, ensuring 
Property~\ref{prop:hitting-set} w.h.p.\ by standard hitting set arguments. 
We call the vertices in $A_i$ the \emph{$i$-landmarks}.
Note that $|A_i|$ is increasing with $i$.
We also define the $0$-landmarks as $A_0 = S$, and there are no $(\DELTA+1)$-landmarks.

\paragraph{Levels.}
The algorithm works in $\DELTA+1$ \emph{levels}, numbered as $i=0,1, \dots, \DELTA$.
Each level only ``cares about'' including the last edges corresponding to triplets $(s,t,c)$ such that the distances between the target $t$ and the vertices of color $c$ are well-structured, as in the following definition.
\begin{definition}[Level-$i$ triplet]
    A triplet $(s,t,c) \in S \times V \times C$ is \emph{in level $i$}, $0 \leq i \leq \DELTA$ if there are $i$ vertices of color $c$ that are $i$-far from $t$, and all other vertices of color $c$ are $(i+1)$-near to $t$. 
\end{definition}

The strategy of dividing into levels indeed suffices for including the last edges for all triplets, as shown by the following observation:

\begin{observation}
    Every triplet $(s,t,c) \in S \times V \times C$ is found in some level $i$, $0 \leq i \leq \DELTA$. 
\end{observation}

\begin{proof}
Let $x_1, \dots, x_\ell$ be the vertices of color $c$, ordered such that $\dist_G (x_1, t) \geq \cdots \geq \dist_G (x_\ell, t)$.
Recall that $\ell \leq \DELTA$.
    Let $i = \min \{ 0 \leq j \leq \ell \mid \text{$j=\ell$, or $\dist_G (x_{j+1}, t) \leq 2d_{j+1}$}\}$.
    Then
    \[
    \dist_G (x_1, t) \geq \cdots \geq \dist_G (x_i, t) > 2d_i > 2d_{i+1} \geq \dist_G (x_{i+1}, t) \geq \cdots \geq \dist_G (x_\ell, t),
    \]
    so $x_1, \dots, x_i$ are $i$-far from $t$, and $x_{i+1}, \dots, x_\ell$ are $(i+1)$-near to $t$.
    So, $(s,t,c)$ is in level $i$.
\end{proof}

The main tool we use to exploit the nice distances structure of $i$-level triplets is formalized in the following~\Cref{lem:level-i}; while its proof is simple, it plays a crucial role in our algorithm.
Intuitively, it allows us to convert a level-$i$ triplet $(s,t,c)$ into a corresponding \emph{highly-structured} triplet $(z,t,c)$ where $z \in A_i \cup S$, such that $\LastE(s,t \mid c) = \LastE(z,t \mid c)$.
Thus, we think of $A_i \cup S$ as the sources for level $i$, or \emph{$i$-sources} for short.
Instead of treating level-$i$ triplets $(s,t,c)$ directly, we will treat all those highly-structured triplets $(z,t,c) \in (A_i \cup S) \times V \times C$ and aim to include the last edges corresponding to them.

While this approach requires us to consider a larger number of sources ($A_i \cup S$ instead of $S$), it comes with major benefits reflected in the structure.
First, for a given $i$-source $z$, we will only care about targets $t$ lying inside the ball $B_G (z, d_i)$.
Second, we will also only care about colors that appear at most $\DELTA-i$ times in this ball $B_G (z,d_i)$, so the ``effective size'' of the faulty set of vertices gets smaller.
Furthermore, all these appearances of color $c$ are \emph{concentrated} in a much smaller ball, of radius roughly $d_{i+1}$, around the target.

\begin{lemma}\label{lem:level-i}
    Let $(s,t,c) \in S \times V \times C$ in level $i$, $0 \leq i \leq \DELTA$.
    Then,
    there is an $i$-source $z \in A_i \cup S$ such that:

    \begin{enumerate}[label=(\arabic*), itemsep=0pt]
        \item $z$ lies on the $d_i$-suffix of $\pi(s,t \mid c)$, and in particular $\LastE(s,t \mid c) = \LastE(z, t \mid c)$.

        \item If $x \in B_G (z, d_i)$ has color $c$, then $x$ is $(i+1)$-near to $t$.

        \item If $x,x' \in B_G (z, d_i)$ have color $c$, then $\dist_G (x,x') \leq 4d_{i+1}$.

        \item The number of $c$-colored vertices in $B_G (z, d_i)$ is at most $\DELTA-i$.
    \end{enumerate}

\end{lemma}
\begin{proof}
    If $\dist_{G-c} (s,t) \leq d_i$ we take $z = s$, and otherwise we can take some $z \in A_i$ found on the $d_i$-suffix of $\pi(s,t \mid c)$, which exists by property~\ref{prop:hitting-set}.  
    Then (1) clearly holds.
    For (2), observe that $\dist_G (x,t) \leq \dist_G (x,z) + \dist_G (z,t) \leq 2d_i$, so $x$ is $i$-near to $t$.
    Since $(s,t,c)$ is in level $i$,
    any $c$-colored vertex that is not $i$-far from $t$ must be $(i+1)$-near to $t$, hence $x$ is $(i+1)$-near to $t$. 
    Item (3) follows immediately from (2) by triangle inequality.
    Finally, (4) follows from (2): indeed, because $(s,t,c)$ is in level $i$, there are $i$ vertices of color $c$ that are $i$-far from $t$.
    As the $c$-colored vertices in $B_G(z,d_i)$ are $(i+1)$-near to $t$, they can be at most $\DELTA-i$.
\end{proof}

We are now ready to give an overview of the algorithm.
First, we explain how the last level is easily dealt with, and afterwards we only focus on the other levels.
As a first step, we give a simple algorithm that achieves near-optimal size for the output, but has suboptimal running time.
Then, we explain the high-level ideas for improving the running time to near-optimal.

\paragraph{Handling the Last Level.}
Item 4 in \Cref{lem:level-i} makes the last level, i.e., $i=\DELTA$, particularly easy to deal with:
We simply compute the SSSP (single-source shortest path) trees rooted at each $\DELTA$-source $z \in A_{\DELTA} \cup S$, and include all these trees in the output $H$.
This adds $O(n \cdot |A_{\DELTA} \cup S|)$ edges and takes $O(m \cdot |A_{\DELTA} \cup S|)$ time, which is easily verified to be within our budget.
To see the correctness, consider some level-$\DELTA$ triplet $(s,t,c)$.
By~\Cref{lem:level-i}, there is some $z \in A_{\DELTA} \cup S$ on the $d_{\DELTA}$-suffix of $\pi(s,t \mid c)$ (item 1), and $B_G (z, d_{\DELTA})$ contains no $c$-colored vertices (item 4).
Hence, the $d_{\DELTA}$-length suffix of $\pi(s,t \mid c)$ is simply $\pi(z,t)$, which is the $z \rsquig t$ path in the tree rooted at $z$.

\paragraph{Achieving Near-Optimal Size.} 
Let us fix some non-last level $0 \leq i \leq \DELTA-1$.
As in the last level, we still compute SSSP trees from the $i$-sources as a preprocessing step:
For every $i$-source $z \in A_i \cup S$, we execute SSSP from $z$ in $G$, but trim it at distance $d_i$;
The resulting rooted tree $T(z)$ is the shortest path tree from $z$ inside the ball $B_G (z, d_i)$.
We add all trees $T(z)$ to the output $H$.
Again, this adds $O(n \cdot |A_i \cup S|)$ edges and takes $O(m \cdot |A_i \cup S|)$ time, which is within budget.

Now, fix an $i$-source $z \in A_i \cup S$.
By~\Cref{lem:level-i}(1), we only care about adding $\LastE(z,t \mid c)$ for triplets $(z,t,c)$ where $t \in B_G (z, d_i)$, i.e., $t \in T(z)$.
In this case, $\pi(z,t)$, which is the $T(z)$-path from the root $z$ to $t$, is already in $H$.
So, if $\pi(z,t)$ avoids the color $c$, then $(z,t,c)$ is already taken care of.
We therefore assume that $t$ has some ancestor $x$ of color $c$ in $T(z)$.
But by~\Cref{lem:level-i}(2) we have $\dist_G (x,t) \leq 2d_{i+1}$, meaning $x$ is in fact within the $2d_{i+1}$ \emph{nearest} ancestors to $t$.

In other words, to treat the $i$-source $z$, it's enough to add $\LastE(z, t \mid c)$ whenever the target $t$ sees the color $c$ among its $2d_{i+1}$ nearest ancestors in $T(z)$.
Thus, each target $t$ only causes the insertion of up to $2d_{i+1}$ edges when processing a specific $i$-source $z$, hence $|A_i \cup S| \cdot 2d_{i+1}$ edges overall (in level $i$).
Summing over $t \in V$, the total number of added edges in level $i$ is at most $n \cdot |A_i \cup S| \cdot 2d_{i+1} = \tilde{O}(n^2 \cdot \frac{d_{i+1}}{d_i}) =  \tilde{O} (n^{2-\frac{1}{\DELTA+1}} |S|^{\frac{1}{\DELTA+1}} )$, as required.

By the discussion above, we get the following simple algorithm for level $i$:
For each $i$-source $z \in A_i \cup S$ and color $c$, run SSSP from $z$ in $G-c$, but add to $H$ only the last edges of the shortest paths $\pi(z,t \mid c)$ such that $t$ has the color $c$ among its $2d_{i+1}$ nearest ancestors in $T(z)$.
This simple algorithm is already enough to ensure the output $H$ is a CFT $S \times V$ distance preserver of the required size.
However, the running time is $O(m \cdot |A_i \cup S| \cdot n)$, which in level $\DELTA-1$ becomes $\tilde{O}(m n^{2-\frac{2}{\DELTA+1}} |S|^{\frac{2}{\DELTA+1}})$.
This is larger than our desired running time by a factor of $n^{1-\frac{1}{\DELTA+1}}|S|^{\frac{1}{\DELTA+1}}$.

\paragraph{Achieving Near-Optimal Time.}
The main challenge lies in improving the running time to near-optimal.
Roughly, the idea is to harness the same argument we used for bounding the \emph{size} of the output $H$, to also effectively bound the \emph{running time}: if we could restrict each edge to participate in only $\tilde{O}(d_{i+1})$ SSSP computations from a specific $i$-source $z$, we would get total running time of $\tilde{O}(m \cdot |A_i \cup S| \cdot d_{i+1}) = \tilde{O}(m n^{1-\frac{1}{\DELTA+1}} |S|^{\frac{1}{\DELTA+1}})$ which is precisely our budget.
Intuitively, to achieve this, we would like the SSSP computation from $z$ in $G-c$ to only use edges incident on vertices $t$ such that the color $c$ appears on their nearest $\tilde{O}(d_{i+1})$ ancestors; these ``allowed'' edges are the only ones we might actually need to add into $H$ as a result of this computation.
Clearly, this intuition doesn't work as is: e.g., it could be that all edges incident on $z$ are not allowed to be used during the SSSP computation from $z$, which is absurd.

To make it work, instead of executing the SSSP directly on $G-c$, we construct a \emph{shortcutted graph} $\hat{G}(z,c)$ that has the allowed original edges, plus extra weighted \emph{shortcut edges} which represent safe paths avoiding the color $c$.

We put two types of shortcuts in $\hat{G}(z,c)$.
The simpler ones represent paths in $T(z)$ that are unaffected by the failure of color $c$. 
Roughly speaking, these simple shortcuts let the SSSP computation ``get inside'' the subtrees of the $c$-colored vertices.
However, the allowed edges are found only in the \emph{shallow} layers of these subtrees (the first $\tilde{O}(d_{i+1})$ layers), while shortest paths from $z$ avoiding $c$ may need to travel into \emph{deep} layers, even if they end in the shallow ones.

The second, more complicated type of shortcuts is designed 
to represent such traversals in the deep layers.
These use the fact that the $(i+1)$-landmarks $A_{i+1}$ must appear on the path at least once every $d_{i+1}$ steps.
The idea is to create a \emph{shell} of additional $\tilde{O}(d_{i+1})$ layers beneath the original shallow layers, which we now call the \emph{kernel} layers.
We also allow the edges touching the shell to be used, which is within budget.
Now, to go from between the kernel layers and the deep layers beneath the shell, one must walk through the shell for at least $d_{i+1}$ steps, and hit some $(i+1)$-landmark there.
Thus, we can add shortcuts to the $(i+1)$-landmarks in the shell, which represent paths that entirely avoid the kernel (and hence the color $c$); these take care of representing the traversal into deep layers.

As it turns out, the shortcut edges are few enough so that they do not hinder the overall running time of the SSSP computations in the shortcutted graphs $\hat{G}(z,c)$.
However, the second-type shortcuts present a problem: how do we compute their weights in the first place?
We need those weights just to \emph{construct} $\hat{G}(z,c)$.
A priori, this requires computing distance in $G$ minus the kernel layers, which \emph{depend on the color} $c$, and is therefore too costly.
We overcome this obstacle by using (a variant of) the well-known tool of \emph{sparse neighborhood covers}.
Essentially, each cluster in the neighborhood cover will be used to define the kernel layers and shell layers for many different colors $c$, and thus we can compute the shortcuts for avoiding the kernels of the different clusters (instead of different colors), which can be done efficiently.

\subsection{Algorithm Description for Level $i\leq \DELTA-1$}

From now on, we focus on the algorithm for level $i$, which ensures that $\LastE(s,t \mid c)$ is included in $H$ for every level-$i$ triplet $(s,t,c)$.
In its analysis, we will show that it takes $\tilde{O}(mn^{1-\frac{1}{\DELTA+1}}|S|^{\frac{1}{\DELTA+1}})$ time while adding $\tilde{O}(n^{2-\frac{1}{\DELTA+1}}|S|^{\frac{1}{\DELTA+1}})$ edges to $H$, so summing over all $\DELTA = O(\log n)$ levels yields the required running time and size bounds of the entire algorithm.

\paragraph{Step 1: Shortest-Path Trees from $i$-Sources.}
For each $i$-source $z \in A_i \cup S$, we execute SSSP in $G$ from the root $z$.
Let $T(z)$ be the resulting shortest-path tree, but trimmed at depth $d_i$.
We add all trees $T(z)$ to the output $H$.
This adds $O(n \cdot |A_i \cup S|)$ edges and takes $O(m \cdot |A_i \cup S|)$ time, which is easily verified to be within our budget.
Note that this step also yields the distance $\dist_G (z, v)$ between any $i$-source $z \in A_i \cup S$ and any vertex $v$.
We will later use these distances as weights for shortcuts in the shortcutted graphs, as explained in the overview.

\paragraph{Step 2: Neighborhood Cover.}
An \emph{$r$-neighborhood cover} is a collection of vertex-subsets called \emph{clusters}, such that every vertex has its $r$-radius ball 
contained in 
some cluster.
The quality of a neighborhood cover is measured by two parameters that one wishes to minimize:
The \emph{diameter} of the clusters, and their \emph{overlap} which is the maximum number of clusters having some vertex common to all of them.
For our purposes, we need a variant of neighborhood cover, where each cluster consists of an inner kernel and shell, and the kernels themselves already cover all the $r$-neighborhoods.
Formally, we use the following result:

\begin{lemma}\label{thm:nbr-cover-with-sep}
    Let $r \geq 1$ be an integer.
    There is an $\tilde{O}(m)$ time algorithm that computes a collection of \emph{kernels} $K_1, 
    \dots, K_\ell \subseteq V$
    with corresponding \emph{clusters} $X_j := B_G (K_j, r)$ and \emph{shells} $Y_j := X_j - K_j$, such that:

    \begin{itemize}[itemsep=0pt]
        \item \emph{(Covering)} Each $v \in V$ has a \emph{covering kernel} $K_j$ with $B_G (v, r) \subseteq K_j$.

        \item \emph{(Diameter)} Each cluster has \emph{weak diameter}%
        \footnote{A vertex set $W \subseteq V$ is said to have weak diameter $d$ if $\dist_G (w,w') \leq d$ for every $w,w' \in W$.}
        $O(r \log n)$.

        \item \emph{(Overlap)} Each $v \in V$ belongs to $O(\log n)$ clusters.
    \end{itemize}
\end{lemma}
\Cref{thm:nbr-cover-with-sep} follows as a corollary from the classical near-linear time algorithm for neighborhood covers (without kernels) of~\cite{ABCP98}. We provide the proof in~\Cref{sec:neighborhood-cover}.

For level $i$, we apply~\Cref{thm:nbr-cover-with-sep} with $r := 4d_{i+1}$, where this choice of radius stems from item 3 of~\Cref{lem:level-i}.
From now on, $K_1, \dots, K_\ell$, $X_1, \dots, X_\ell$ and $Y_1, \dots ,Y_\ell$ denote the kernels, clusters and shells resulting from this application.

As explained in the overview, the shortcutted graphs will have weighted shortcut edges based on the neighborhood cover.
These shortcuts will correspond to distances from the following set:
\[
\big\{\dist_{G-K_j} (a,v) \mid 1\leq j \leq \ell, ~ a \in A_{i+1} \cap Y_j, ~ v \in V-K_j \big\}.
\]
Namely, for each $(i+1)$-landmark $a \in A_{i+1}$ and each $j$ such that $a$ belongs to the shell $Y_j$, we need to compute SSSP in $G-K_j$, i.e., avoiding the kernel $K_j$.
By the overlap guarantee in~\Cref{thm:nbr-cover-with-sep}, each $(i+1)$-landmark appears in $O(\log n)$ shells, and hence triggers only $O(\log n)$ SSSP calls.
Thus, the total time spent for computing these distances is $\tilde{O}(m \cdot |A_{i+1}|) \leq \tilde{O}(m \cdot |A_{\DELTA}|) = \tilde{O}(m \cdot n^{1 - \frac{1}{\DELTA+1}} \cdot |S|^{\frac{1}{\DELTA+1}})$, which is within budget.

\paragraph{Step 3: Finding Relevant Colors.}
We also use the neighborhood cover to define the notion of \emph{relevant colors}:

\begin{definition}[Relevant Color for $i$-Source]\label{def:relevant-color}
Color $c$ is \emph{relevant} for $i$-source $z \in A_i \cup S$ if there exists some $c$-colored vertex $x \in B_G (z, d_i)$ whose covering kernel $K_j$ (i.e., the kernel with $B_G (x, 4d_{i+1}) \subseteq K_j$) contains \emph{all} $c$-colored vertices in $B_G (z,d_i)$.
In this case, we choose one such arbitrary $x$, and denote the kernel, cluster and shell for the pair $z,c$ by $K(z,c) := K_j$, $X(z,c) := X_j$ and $Y(z,c) := Y_j$, respectively.
\end{definition}

The reasoning behind the name ``relevant'' stems from~\Cref{lem:level-i}(3):
it says that an $i$-source $z$ should only care about colors $c$ such that the $c$-colored vertices in $B_G (z, d_i)$ are within distance $4d_{i+1}$ from each other, and hence the covering kernel of any one of them would contain all of them.
To check if some color $c$ is relevant for some $i$-source $z$ only takes $\tilde{O}(1)$ time, since each vertex in $V_c \cap B_G (z,d_i)$ only belongs to $O(\log n)$ kernels, and $|V_c| \leq \DELTA = O(\log n)$.
Thus, computing all pairs of $i$-sources and relevant colors takes time $\sum_{z \in A_i \cup S} \tilde{O}(|B_G(z, d_i)|) \leq |A_\DELTA \cup S| \times \tilde{O}(m) = \tilde{O}(m \cdot n^{1 - \frac{1}{\DELTA+1}} \cdot |S|^{\frac{1}{\DELTA+1}})$, which is within budget.

\paragraph{Step 4: Constructing Shortcutted Graphs.}
Fix an $i$-source $z \in A_i \cup S$ and relevant color $c$.
We now give the construction of the corresponding shortcutted graph $\hat{G}(z,c)$.
We first introduce some notations and terminologies.
To avoid clutter, in the following we may omit $z,c$ from the notations when they are clear from context.
\begin{itemize}[itemsep=0pt]
    \item \emph{The ball} refers to $B := B_G (z,d_i)$, and \emph{the tree} refers to the shortest paths tree $T := T(z)$ from $z$ in $B$, which was computed in Step 1.

    \item A vertex $v \in B$ is called \emph{affected} if it is a descendant of some $c$-colored vertex in the tree $T$, and the set of affected vertices is denoted by $D := D(z,c)$. 

    \item The \emph{kernel}, \emph{cluster}, and \emph{shell} refer to $K := K(z,c)$, $X := X(z,c)$, and $Y := Y(z,c)$ from~\Cref{def:relevant-color}, respectively.
    Recall that the cluster is the $4d_{i+1}$-radius ball around the kernel, i.e., $X = B_G (K, 4d_{i+1})$, and the shell is the part of the cluster outside the kernel, i.e., $Y = X-K$.
    
    \item The \emph{affected kernel}, \emph{affected cluster} and \emph{affected shell} are the intersections of the kernel, cluster and shell with the affected set $D$, 
    denoted $K_D$, $X_D$ and $Y_D$
    respectively.

    \item The \emph{designated landmarks} $L := L(z,c)$ are all those $(i+1)$-landmarks found in the affected shell. That is, $L = A_{i+1} \cap Y_D$.

    \item The \emph{outer layer} $O := O(z,c)$ are all those vertices in $B$ that have some neighbor from the affected cluster, but are not affected themselves. That is, $O = N_G (X_D) \cap (B-D)$.
\end{itemize}

We define the shortcut graph for $z$ and $c$ as follows (see verbal description following the formal definition below).
\begin{mdframed}[style=MyFrame, nobreak=true]
\begin{center}
\textbf{The Shortcutted Graph $\hat{G}(z,c)$}
\end{center}
\medskip
\noindent For $z \in A_i \cup S$ and relevant color $c$, the shortcutted graph $\hat{G} = \hat{G}(z,c)$ is defined by
\begin{align*}
    V(\hat{G}) &:= B, \\
    E(\hat{G}) &:= \operatorname{Original-Edges}(\hat{G}) \cup \operatorname{Landmark-Shortcuts}(\hat{G})
    \cup \operatorname{Root-Shortcuts}(\hat{G}),
\end{align*}
where the edges and their weights are defined by
\begin{align*}
    \operatorname{Original-Edges}(\hat{G}) &:= \{ \text{$(u,v)$ of weight $1$} \mid (u,v) \in E, ~ u \in B, ~ v \in X_D\} \\
    \operatorname{Landmark-Shortcuts}(\hat{G}) &:= \{\text{$(a,v)$ of weight $\dist_{G-K} (a,v)$} \mid a \in L, ~ v \in B-K\} \\
    \operatorname{Root-Shortcuts}(\hat{G}) &:=
    \begin{cases}
    \text{if $L \neq \emptyset$:} & \{(z,v) \text{ of weight } \dist_G (z,v) \mid v \in O \text{ or } v \in B-(K \cup D) \} \\    
    \text{if $L = \emptyset$:} & \{(z,v) \text{ of weight } \dist_G (z,v) \mid v \in O\}
    \end{cases} \\
\end{align*}
\end{mdframed}

The vertex set of $\hat{G}$ is just the ball $B$.
We include in $\hat{G}$ every \emph{original edge} from $G$ that has some endpoint in the affected cluster $X_D$ (and the other endpoint can be anywhere in $B$).
These original edges are given weight $1$.
Then, we add two types of \emph{shortcut edges} to $\hat{G}$.
One type is the \emph{landmark shortcuts}:
From each designated landmark $a \in L$, we add edges to every vertex $v$ outside the kernel $K$, with weight $\dist_{G-K} (a,v)$.
The other type is the \emph{root shortcuts}:
First, we add edges from the root $z$ to every vertex $v$ in the outer layer $O$, with weight $\dist_G (z,v)$.
Additionally, in the ``non-degenerate'' case where $L \neq \emptyset$, i.e., there is \emph{some} designated landmark, we add such root shortcuts from $z$ to every $v \in B-(K \cup D)$, i.e., $v$ is a non-affected vertex outside the kernel.
(The reason behind this slight complication is rather technical. Essentially, if $L = \emptyset$, then the latter kind of root shortcuts are not required for us, but including them might damage the running time.)

We defer the analysis of the time it takes to construct the shortcut graphs, as it is more convenient to analyze together with the time to execute the next Step 5.

\paragraph{Step 5: SSSPs in Shortcutted Graphs.}
For every pair of $i$-source $z \in A_i \cup S$ and relevant color $c$, we consider the corresponding shortcutted graph $\hat{G} = \hat{G}(z,c)$.
We execute (weighted) SSSP from $z$ in $\hat{G}-c$.
Then, for each target vertex $t$ in the affected kernel $K_D$, we add the last edge of the $z \rsquig t$ shortest path computed by this SSSP into the output $H$.
Note that this edge must be original, since shortcuts do not touch the affected kernel.

\medskip
This concludes the description of the algorithm for level $i$.
Next, we provide an analysis of the output size and running time (for steps 4 and 5).
Afterward, we provide the correctness proof for the algorithm.

\subsection{Size and Running Time Analysis}
The analysis of the output size and the algorithm's running time hinges on the following:
\begin{lemma}\label{lem:affected-clusters}
    Let $z \in A_i \cup S$ and color $c$ relevant for $z$.
    Then, every vertex in the affected cluster $X_D (z,c)$ must have a $c$-colored vertex among its $O(d_{i+1} \log n)$  nearest ancestors in $T(z)$.
\end{lemma}
\begin{proof}
    Let $v \in X_D (z,c)$.
    Then, because $v$ is affected, it has some $c$-colored ancestor $x$ in $T(z)$.
    Since $c$ is a relevant color for $z$, the corresponding cluster $X(z,c)$ (in fact, even its kernel $K(z,c)$) contains every $c$-colored vertex in $T(z)$, and particularly $x$.
    As clusters have weak diameter $O(d_{i+1} \log n)$, the distance between $x$ and $v$ in $G$, and hence also in the shortest-path tree $T(z)$, is $O(d_{i+1} \log n)$.
\end{proof}

First, we analyze the number of edges added to the output $H$ during level $i$.
Each subroutine for pair $z,c$ adds as many edges as the size of the affected kernel.
So, by~\Cref{lem:affected-clusters}, the total number of added edges during level $i$ is
$|V| \cdot |A_i \cup S| \cdot O(d_{i+1} \log n) = \tilde{O} (n^2 \cdot \frac{d_{i+1}}{d_i}) = \tilde{O}(n^{2-\frac{1}{\DELTA+1}} |S|^{\frac{1}{\DELTA+1}})$, 
which is within our budget.

We now focus on the running time.
To start, let us consider the time required to \emph{construct} the shortcutted graphs.
The primary task here is to identify the vertices in the affected kernel/cluster/shell.
By~\Cref{lem:affected-clusters}, this can be done by scanning the subtrees rooted in the $c$-colored vertices in $T(z)$, but only up to depth $O(d_{i+1} \log n)$ from the root of the subtree, and checking if the scanned vertices belong to the corresponding kernel/cluster/shell.
Thus, fixing $z$ and letting $c$ vary, each edge in $T(z)$ is scanned only $O(d_{i+1} \log n)$ times.
Now summing over $z$, we bound the running time for identifying all affected clusters by $m \cdot |A_i \cup S| \cdot O(d_{i+1} \log n) = \tilde{O}(mn\frac{d_{i+1}}{d_i}) = \tilde{O}(mn^{1-\frac{1}{\DELTA+1}}|S|^{\frac{1}{\DELTA+1}})$, which is within budget.
Once the affected kernel/cluster/shell (and the designated landmarks inside the affected shell) are identified, the rest of the time to construct the shortcutted graph $\hat{G}(z,c)$ is linear in the number of edges (original or shortcuts) that $\hat{G}(z,c)$ has.
The SSSP we execute in this $\hat{G}(z,c)$ is also near-linear in the number of edges; hence, it remains to account for the time required to run the SSSPs in the shortcutted graphs.

Observe that the number of edges in $\hat{G}(z,c)$ is dominated by its original edges and its landmark shortcuts.
Indeed, a root shortcut of the form $(z,v)$ with $v \in O$ can be charged to some original edge $(v,u)$ with $u \in X_D$.
Also, a root shortcut of the form $(z,v)$ with $v \in B-(K \cup D)$ can be charged to a landmark shortcut $(a,v)$ with $a \in L$.
Now, recall that original edges and landmark shortcuts have at least one endpoint in the affected cluster.
Thus, by~\Cref{lem:affected-clusters}, each original edge $e \in E$ can appear in at most $|A_i \cup S| \cdot O(d_{i+1} \log n) = \tilde{O}(n \frac{d_{i+1}}{d_i}) = \tilde{O}( n^{1 - \frac{1}{\DELTA+1}} |S|^{\frac{1}{\DELTA+1}} )$ shortcutted graphs.
Similarly, each pair in $A_{i+1} \times V$ can appear as a landmark shortcut (with varying weight) in $\tilde{O}( n^{1 - \frac{1}{\DELTA+1}} |S|^{\frac{1}{\DELTA+1}} )$ shortcutted graphs.
So, summing over all pairs $z,c$, the total time for level $i$ (other than the preprocessing, which we already accounted for) is
\[
(m + |A_{i+1}| \cdot n) \cdot \tilde{O}(n^{1-\frac{1}{\DELTA+1}} |S|^{\frac{1}{\DELTA+1}})
\]
Note that
$|A_{i+1}| \cdot n \leq |A_{\DELTA}| \cdot n = \tilde{O}(n^{2-\frac{1}{\DELTA+1}} |S|^{\frac{1}{\DELTA+1}}) \leq \tilde{O}(m)$,
so the total running time of level $i$ is
$
\tilde{O}(m n^{1-\frac{1}{\DELTA+1}} |S|^{\frac{1}{\DELTA+1}})
$
as needed.

\subsection{Correctness}
We now prove the correctness of the algorithm in level $i$, meaning that it includes in the output $H$ all of the required last edges in this level.
Let $(s,t,c) \in S \times V \times C$ be a level-$i$ triplet; our goal is to show that $\LastE(s,t \mid c) \in H$.
To this end, let $z \in A_i \cup S$ be the $i$-source from~\Cref{lem:level-i} for the triplet $(s,t,c)$.
By~\Cref{lem:level-i}(1), we have $\LastE(s,t \mid c) = \LastE(z,t \mid c)$, and we will show that the latter has been added to $H$ during level $i$.

We first handle a trivial case, when $\pi(z,t)$ doesn't have the color $c$; then $\LastE(z,t 
\mid c)$ is just the last edge of $\pi(z,t)$.
Note that $\pi(z,t)$ has length at most $d_i$ by~\Cref{lem:level-i}(1), and hence it is found in the shortest-path tree $T = T(z)$ of the ball $B = B_G (z,d_i)$, which we've added to $H$.
Thus, we may assume $t$ has some $c$-colored ancestor $x$ in the tree $T$.
By~\Cref{lem:level-i}(3) we have $V_c \cap B \subseteq B_G (x, 4d_{i+1})$, hence the covering kernel of $x$ contains all of $V_c \cap B$.
Therefore, color $c$ is relevant for the $i$-source $z$.
We thus focus on the subroutine executed for the pair $z,c$ (and again, omit them from notations to avoid clutter).

Consider $K$, $X$, and $Y$, which are, respectively, the kernel, cluster, and shell we have chosen for $z$ and $c$.
By their definition, there is some $c$-colored vertex $x' \in B$ such that $B_G (x', 4d_{i+1}) \subseteq K$.
Now, by~\Cref{lem:level-i}(2), we have $\dist_G (x',t) \leq 2d_{i+1}$, implying that $t \in K$.
Recall that $t$ also has the $c$-colored vertex $x$ as an ancestor, meaning $t$ is an affected vertex in $T$, i.e., $t \in D$.
Thus, $t$ is in the affected kernel $K_D$, meaning that the subroutine for $z,c$ included in $H$ the last edge of the shortest $z \rsquig t$ path in $\hat{G}-c$, and our goal is to show that this is in fact $\LastE(z,t \mid c)$.
To this end, we first observe the following:
\begin{observation}
The following hold:
\begin{enumerate}[label=(\arabic*), itemsep=0pt]
    \item $\pi(z,t \mid c)$ is the shortest $z \rsquig t$ path in $G-(V_c \cap B)$ (and not only in $G-c$).
    
    \item Every path in $\hat{G}-c$ corresponds to some path in $G-(V_c \cap B)$.
\end{enumerate}
\end{observation}
\begin{proof}
    (1) By~\Cref{lem:level-i}(1), $\dist_{G-c}(z,t) \leq d_i$, so $\pi(z,t \mid c)$ cannot leave the $d_i$-ball $B$ around $z$.
    
    (2) This follows as shortcuts in $\hat{G}$ correspond to paths in $G - (V_c \cap B)$.
    Indeed, root shortcuts correspond to paths in $T$ to unaffected vertices, which avoid the color $c$ entirely.
    As for landmark shortcuts, they represent paths in $G-K$, and the kernel $K$ contains $V_c \cap B$ by its definition.
\end{proof}

Thus, the heart of the argument lies in the following lemma, completing the correctness proof:
\begin{lemma}
    There is a path in $\hat{G}-c$ that corresponds to $\pi(z,t \mid c)$.
\end{lemma}
\begin{proof}
    First, observe that once $\pi(z, t \mid c)$ enters the affected set $D$, it stays in $D$.
    This is because if $v \in B-D$ is some unaffected vertex, then the shortest $z \rsquig v$ in $G-c$ is simply the tree path in $T$, which only consists of unaffected vertices.
    We use this observation (implicitly) throughout.

    Partition $\pi(z, t \mid c)$ into subpaths $P_1 \circ P_2 \cdots \circ P_{2q}$ (last vertex of $P_j$ is first vertex of $P_{j+1}$), where the even subpaths $P_2, P_4, \cdots, P_{2q}$ are exactly the maximal subpaths within the affected core $K_D$.
    (This is well defined as $z \notin K_D$ and $t \in K_D$.)
    All original edges inside the affected cluster exist in $\hat{G}$, so the even subpaths are taken care of, and we just need to argue about the odd ones.

    Consider any odd subpath $P_j$ except the prefix $P_1$.
    If $P_j$ does not leave the affected cluster $X_D$, then $P_j$ exists in $\hat{G}$ as is.
    Otherwise, $P_j$ must take its first edge from some $u \in K_D$ into some $v$ in the affected shell $Y_D$, and walk through $Y_D$ for at least $4d_{i+1}$ steps to leave the cluster (since the shell is the $4d_{i+1}$-width ring around the kernel).
    So by the hitting set property~\ref{prop:hitting-set}, this traversal through the shell must hit some $(i+1)$-landmark $a \in A_{i+1} \cap Y_D = L$.
    From $a$, $P_j$ continues outside the kernel up until its last edge connecting some $v' \in Y_D$ to the final vertex $u' \in K_D$.
    Note that the edges $(u,v)$ and $(v',u')$ are within the affected cluster $X_D$, and hence included in $\hat{G}$.
    Also, $\hat{G}$ has the landmark shortcuts $(v,a)$ and $(a,v)$ with weights $\dist_{G-K}(v,a)$ and $\dist_{G-K}(a,v)$, which is precisely the distances that $P_j$ travels between $v,a$, and between $a,v'$.
    Thus, the path $(u,v) \circ (v,a) \circ (a,v') \circ (v',u)$ in $\hat{G}-c$ corresponds to $P_j$.

    It remains to consider the prefix $P_1$.
    Let us first handle the case where $P_1$ does not go through $D-X_D$, i.e., it only visits unaffected vertices or travels inside the affected cluster.
    Let $v$ be the last unaffected vertex on $P_1$.
    The next vertex after $v$ must be in $X_D$, meaning $v \in O$.
    Thus, $\hat{G}$ has the root shortcut $(z,v)$ with weight $\dist_G (z,v)$, which is exactly the distance $P_1$ travels between $z,v$.
    After reaching $v$, the path $P_1$ only uses edges touching $X_D$, which exist in $\hat{G}$.
    So, $P_1$ has a corresponding path in $\hat{G}-c$ as needed.

    Finally, suppose $P_1$ goes through $D-X_D$.
    Then, after visiting its last vertex from $D-X_D$, it must travel for at least $4d_{i+1}$ steps inside the shell $Y_D$, until it takes its final edge from some $v' \in Y_D$ to its final vertex $u' \in K_D$.
    So, by property~\ref{prop:hitting-set}, this traversal on the shell must hit some designated landmark $a \in L$ (and in particular $L \neq \emptyset$).
    Then, $\hat{G}$ has the landmark shortcut $(a,v')$ of weight $\dist_{G-K} (a,v')$, which is precisely the distance that $P_1$ travels between $a,v'$.
    Also, the edge $(v',u')$ is within $X_D$, and thus exists in $\hat{G}$.
    It remains to show that the subpath $P'_1$ of $P_1$ between $z,a$ has a corresponding path in $\hat{G}-c$.
    To this end, consider the last unaffected $u \in B-D$ vertex that $P'_1$ visits.
    \begin{itemize}[itemsep=0pt]
        \item If $u \notin K$, then $\hat{G}$ has the root shortcut $(z,u)$ of weight $\dist_G (z,u)$, which is the distance $P'_1$ travels between $z,u$.
        Also, $\hat{G}$ has the landmark shortcut $(u,a)$ with weight $\dist_{G-K} (u,a)$, which is the distance $P'_1$ travels between $u,a$.

        \item If $u \in K$, then consider the next vertex $v$ after $u$ in $P'_1$.
        Then $v$ is affected, so it cannot be in the kernel (as $P_1$ only reaches the affected kernel in its 
        last vertex $u$).
        Thus, $\hat{G}$ has the landmark shortcut $(v,a)$ of weight $\dist_{G-K}$, which is the distance of $P'_1$ between $v,a$.
        We thus remain with the subpath $P''_1$ between $z,v$.
        As $v$ is an affected neighbor of the kernel, it must be in the affected cluster $X_D$, and hence $u \in O$.
        Thus, the original edge $(u,v)$ is in $\hat{G}$, and $\hat{G}$ contains a root shortcut $(z,u)$ of weight $\dist_G (z,u)$: together, they correspond to $P''_1$.
    \end{itemize}
    So, in both cases $\hat{G}-c$ has a corresponding path to $P'_1$, and we are done.    
\end{proof}

This concludes the proof of~\Cref{thm:alg-upper-bound}.

\section{Lower Bounds for Multi-Source CFT Distance Preservers}

\subsection{Tree Construction}

Both of the lower bounds for CFT multi-source distance preservers (size and running time) hinge on a colored tree construction, which adapts the one given by Parter~\cite{Parter15} for the $\DELTA$-FT setting to hold also in the more relaxed CFT setting with color classes of size at most $\DELTA$.

\begin{lemma}\label{lem:T_Delta_q}
    Let $\DELTA \geq 0$ and $q \geq 2$ be integers. 
    There exists a rooted, partially colored%
    \footnote{``Partially colored'' means that vertices/edges with no color are allowed.
    Alternatively, uncolored vertices/edges can be thought of as having a unique color that does not appear anywhere else.
    }
    tree $T = T(\DELTA,q)$ with the following properties:
    \begin{enumerate}[label=(\alph*), itemsep=0pt]
        \item For every leaf $v$ of $T$, there exists a color $c_v$ such that $v$ is the \textbf{unique} leaf of minimal depth among those leaves that remain connected to the root in $T-c_v$.%
        \label{prop:distances}

        \item Each color appears on at most $\DELTA$ edges in $T$.%
        \label{prop:colors}
        
        \item $T$ has $q^\DELTA$ leaves.%
        \label{prop:leaves}

        \item Each leaf in $T$ has depth at least $d(\DELTA, q) := q^{\DELTA}-1$ and at most $D(\DELTA,q) := 2(q^{\DELTA}-1)$.
        \label{prop:depth}
        
        \item The number of edges in $T$ is $N(\DELTA,q) = \DELTA (q^{\DELTA+1} - q^{\DELTA-1})$.
        \label{prop:edges}

        \item The tree $T$ and the colors $c_v$ corresponding to each leaf $v$ can be computed in linear time (w.r.t.\ the size of $T$).
        \label{prop:linear-time}
    \end{enumerate}
\end{lemma}

\begin{proof}
We show the proof with edge colors: the vertex-colored version of the tree is obtained by assigning each vertex the color of the edge connecting it to its parent.
By induction on $\DELTA$.
For the base case, $T(0,q)$ is a single vertex $v$ with no edges.
We artificially define the color palette of $T(0,q)$ as $\{c_v\}$, with the color $c_v$ unused.
The required properties clearly hold.
For the induction step, we construct $T = T(\DELTA, q)$ for $\DELTA \geq 1$ as follows.
\begin{itemize}[itemsep=0pt]
    \item Create $q$ disjoint copies $T_1, \dots, T_q$ of $T(\DELTA-1, q)$, with disjoint color palettes (so that no color appears in two different copies).
    The leaves of $T$ will be those of $T_1, \dots, T_q$.
    For a leaf $v$ of $T_i$, we define $c_v$ as the same color guaranteed to exist by Property~\ref{prop:distances} of $T_i$.

    \item Create $q$ disjoint uncolored paths $Q_1, \dots, Q_q$, where
    $
    |Q_i| = 2(q-i)q^{\DELTA-1}
    $
    (so $Q_q$ is a single vertex).
    ``Hang'' each $T_i$ at the end of $Q_i$, i.e., the last vertex of $Q_i$ is the root $r_i$ of $T_i$.

    \item Let $u_1, \dots, u_q$ be the first vertices of  $Q_1, \dots, Q_q$ respectively.
    Connect each consecutive pair $u_i, u_{i+1}$ by a path $P_i$ with $|P_i| = q^{\DELTA-1}$, and color $P_i$ as follows:
    for each leaf $v$ of $T_i$, assign its corresponding color $c_v$ to some unique edge in $P_i$.
    (By Property~\ref{prop:leaves},
    $T_i$ has $q^{\DELTA-1}$ leaves, so $P_i$ has exactly enough edges to accommodate the colors.)

    \item Define the root of $T$ as $r := u_1$.
\end{itemize}
See Figure~\ref{fig:T(Delta,q)} for an illustration.
\begin{figure}
    \centering
    \resizebox{8cm}{!}{\input{T_Delta,q_.tikz}} 
    \caption{
        Construction of the tree $T(\DELTA=2, q=3)$.
        $T_1$, $T_2$ and $T_3$ are copies of $T(1,3)$ with disjoint colors.
        The corresponding color $c_v$ for each leaf $v$ is shown beneath it.
    }
    \label{fig:T(Delta,q)}
\end{figure}
To prove the stated properties, we first focus on the paths going from the root of $T$ to the roots of $T_1, \dots, T_q$, denoted by $R_1, \dots, R_q$ respectively.
By Property~\ref{prop:depth} of $T(\DELTA-1,q)$, it holds that
\begin{equation}\label{eq:leaf-depth}
    |R_i| + d(\DELTA-1, q) \leq \depth_T (v) \leq |R_i| + D(\DELTA-1,q) \quad\text{for any leaf $v$ in $T_i$.}
\end{equation}
We calculate $|R_i|$;
observe that
$
R_i = P_1 \circ \cdots \circ P_{i-1} \circ Q_i,
$
so
\begin{align}
    |R_i| 
    &= (i-1)q^{\DELTA-1} + 2(q-i)q^{\DELTA-1} = (2q-i-1)q^{\DELTA-1}
    \label{eq:Ri-length}
\end{align}
We are now ready to show the required properties.
The most interesting is Property~\ref{prop:distances}; the rest follow by straightforward induction.

\begin{itemize}[itemsep=0pt]
    \item[\ref{prop:distances}] 
    Consider a leaf $v$, originating from $T_i$.
    The only path of $P_1, \dots, P_{q-1}$ on which the color $c_v$ appears is $P_i$ (or, if $i=q$, none of them contain $c_v$).
    Hence, the path $R_j$ survives in $T-c_v$ iff $j \leq i$.
    Therefore, a leaf $v'$ of $T$ that remains connected to $r$ in $T-c_v$ is either
        (i) 
        found in $T_i$ and remains connected to its root in $T_i - c_v$, or
        (ii)
        found in some $T_j$ with $j < i$.\\
    By Property~\ref{prop:distances} of $T_i$, $v$ has minimal depth in $T_i$ among the type-(i) leaves, hence this also holds in $T$.
    Next, let $v'$ be a type-(ii) leaf. In this case,
    \begin{align*}
        &\depth_T (v') - \depth_T (v) \\
        &\geq \big(|R_j| - |R_i|\big) + \big(D(\DELTA-1,q) - d(\DELTA-1,q)\big) &&\text{by \Cref{eq:leaf-depth}}\\
        &= (i-j)q^{\DELTA-1} - (q^{\DELTA-1}-1) &&\text{by \Cref{eq:Ri-length}, def.\ of $d(\cdot,q), D(\cdot,q)$} \\
        &\geq 1 &&\text{as $i-j \geq 1$.}
    \end{align*}

    \item[\ref{prop:colors}]
    Each color $c$ in $T$ appears at most $\DELTA-1$ times in some $T_i$ (by induction), at most once in $P_i$, and nowhere else, so overall at most $\DELTA$ times.

    \item[\ref{prop:leaves}]
    By induction, each $T_i$ has $q^{\DELTA-1}$ leaves, hence $T$ has $q \cdot q^{\DELTA-1} = q^{\DELTA}$ leaves.

    \item[\ref{prop:depth}]
    The shortest (longest) of the $R_i$'s is $R_q$ ($R_1$) by~\Cref{eq:Ri-length}.
    So by~\Cref{eq:leaf-depth} for leaf $v$,
    \begin{align*}
        \depth_T (v) &\geq |R_q| + d(\DELTA-1,q) = (q-1)q^{\DELTA-1} + q^{\DELTA-1} - 1 = q^{\DELTA} - 1 = d(\DELTA, q) , \\
        \depth_T (v) &\leq |R_1| + D(\DELTA-1, q) = 2(q-1)q^{\DELTA-1} + 2(q^{\DELTA-1} - 1) =  2(q^{\DELTA}-1) = D(\DELTA, q).
    \end{align*}

    \item[\ref{prop:edges}]
    Each edge of $T$ either is in some $T_i$, $Q_i$ or $P_i$.
    Hence, the total number of edges is
    \begin{align*}
        q \cdot N(\DELTA-1,q)  + \sum_{i=1}^q |Q_i| + \sum_{i=1}^{q-1} |P_i|
        &= q\cdot (\DELTA-1) \cdot (q^{\DELTA} - q^{\DELTA-2}) + (q-1) \cdot q^{\DELTA} + (q-1) \cdot q^{\DELTA-1}  \\
        &= \DELTA \cdot (q^{\DELTA+1} - q^{\DELTA-1}) 
        = N(\DELTA,q).
    \end{align*}

    \item[\ref{prop:linear-time}]
    Immediate by the construction procedure of $T$ from $T_1, \dots, T_q$.

\end{itemize}
\end{proof}

\subsection{Size Lower Bound}

We now prove the tight lower bound on the worst-case size of CFT multi-source distance preservers.

\begin{theorem}\label{thm:size-lower-bound}
    For any integers $1 \leq \DELTA, \sigma \leq n$,
    there exists an $n$-vertex colored graph $G = (V,E)$, where each color appears at most $\DELTA$ times, and sources $S \subseteq V$ with $|S| = \sigma$, such that any CFT $S \times V$ distance preserver of $G$ has at least $\tilde{\Omega}(n^{2-\frac{1}{\DELTA+1}} \cdot \sigma^{\frac{1}{\DELTA+1}})$ edges.
\end{theorem}

\begin{proof}
We construct the lower bound instance as follows:
\begin{itemize}[itemsep=0pt]
    \item Take $\sigma$ disjoint copies $T_1, T_2, \dots, T_\sigma$ of the tree $T(\DELTA, q)$ from~\Cref{lem:T_Delta_q}, where we set $q := \big(\frac{n}{\sigma \DELTA} \big)^{\frac{1}{\DELTA+1}}$.
    For each copy, we use a different color palette, so that the trees are also \emph{color-disjoint}, with color classes of size at most $\DELTA$ by \Cref{lem:T_Delta_q}\ref{prop:colors}.

    \item Define $S$ as the $\sigma$ different roots $s_1, \dots, s_\sigma$ of $T_1, \dots, T_\sigma$ (respectively).

    \item Denote by $X$ the set of leaves of $T_1, \dots, T_\sigma$.

    \item Create a set $Y$ of $n$ new vertices, and add all edges $(x,y)\in X \times Y$ (with no color) to $G$.
\end{itemize}
The total number of vertices is
\[
|V(G)| = n + \sum_{i=1}^\sigma |V(T_i)| = n + \sigma \cdot O(\DELTA q^{\DELTA+1}) = n + \sigma \cdot O\big(\frac{n}{\sigma}\big) = \Theta(n),
\]
where the first equality is by \Cref{lem:T_Delta_q}\ref{prop:edges}, and the second by choice of $q$.
The number of edges in $X \times Y$ is 
\[
n \cdot \sum_{i=1}^\sigma \#\text{leaves in $T_i$} = n \cdot \sigma q^\DELTA = \Omega \big( n^{2-\frac{1}{\DELTA+1}} \cdot \sigma^{\frac{1}{\DELTA+1}} \cdot \tfrac{1}{\DELTA} \big),
\]
where the first equality is by \Cref{lem:T_Delta_q}\ref{prop:leaves}, and the second by choice of $q$.
So, it suffices to show that a CFT $S \times V$ distance preserver $H$ of $G$  must contain every edge $(x,y) \in X \times Y$.

Consider the copy $T_i$ of $T(\DELTA,q)$ in which $x$ is a leaf.
By \Cref{lem:T_Delta_q}\ref{prop:distances}, there is a color $c_x$ such that $x$ is the unique leaf of minimal depth in $T_i$ that remains connected to $s_i$ in $T_i - c_x$.
By construction of $G$, any $s_i \rsquig y$ path in $G-c_x$ must start with an $s_i$-to-leaf path in $T_i - c_x$, and take at least one more edge to arrive at $y$.
It follows that there is a \emph{unique} shortest $s_i \rsquig y$ path in $G-c_x$, which goes along the $s_i \rsquig x$ path in $T_i - c_x$, then takes the edge $(x,y)$.
As $H$ is a subgraph of $G$ and $\dist_{H-c} (s,y) = \dist_{G-c} (s,y)$, it must contain this path, and particularly the edge $(x,y)$.
\end{proof}

\subsection{Conditional Lower Bound for Combinatorial Algorithms}

Finally, we give the lower bound for combinatorial algorithms computing CFT multi-source distance preservers, conditioned on the following:

\begin{conjecture}[see e.g.~\cite{RodittyZ11,WilliamsW18,AbboudW14,ChechikC19}]\label{conj:BMM}
    Any combinatorial algorithm that gets as input two Boolean $n \times n$ matrices $A, B \in \{0,1\}^n$ with a total of $m$ $1$'s and outputs their (OR,AND)-product $AB$, requires $\Omega((mn)^{1-o(1)} + n^2)$ running time.
\end{conjecture}

\begin{theorem}\label{thm:alg-lower-bound}
    Fix $\DELTA \geq 1$ and $a,b \in [0,1]$ such that \ $a > \frac{\DELTA+b}{\DELTA+1}$.
    For any integer $n \geq 1$, denote $m = m(n) = n^{1+a+o(1)}$ and $\sigma = \sigma(n) = n^{b+o(1)}$.
    (So, $m$ is polynomially larger than $n^{2-\frac{1}{\DELTA+1}} \cdot \sigma^{\frac{1}{\DELTA+1}}$.)
    Suppose there is a combinatorial algorithm with the following guarantee:
    There is a constant $\epsilon > 0$ such that given a colored graph with color classes of size at most $\DELTA$ that has $n$ vertices, $m$ edges and $\sigma$ sources, the algorithm outputs a CFT multi-source distance preserver with $m^{1-\epsilon}$ edges within $n^{2 + a -\frac{1-b}{\DELTA+1} - \epsilon}$ time (i.e., polynomially faster than $m \cdot n^{1-\frac{1}{\DELTA+1}} \cdot \sigma^{\frac{1}{\DELTA+1}}$).
    Then~\Cref{conj:BMM} is false.
\end{theorem}

\begin{proof}
Consider a BMM instance, i.e., two $n \times n$ Boolean matrices $A,B$ with a total of at most $m = n^{1+a + o(1)}$ $1$'s.
We represent this instance as a graph $G$ whose vertices are partitioned to three sets $I,J,K$ of size $n$, where there is an edge between $i \in I$ and $j \in J$ iff $A[i,j] = 1$, and similarly, there is an edge between $j \in J$ and $k \in K$ iff $B[j,k] = 1$.
Thus, for $i \in I$ and $k \in K$, there is a $2$-path between $i,k$ iff the Boolean product of $A,B$ has a $1$ in entry $i,k$.
So, to solve the BMM instance, we need to determine which pairs in $I \times K$ have $2$-paths.
Recall that $\sigma = n^{b+o(1)}$.
We partition $I$ into $r := (n/\sigma)^{\frac{1}{\DELTA+1}}$ blocks of size $n/r = n^{\frac{\DELTA}{\DELTA+1}} \sigma^{\frac{1}{\DELTA+1}}$.
The algorithm works in a block-by-block manner; from now on, we focus on finding the pairs with 2-paths in $I' \times K$ for some block $I' \subset I$. 
We further partition $I'$ into $\sigma$ sub-blocks $I'_1, \dots, I'_\sigma$ of size $n/(r \sigma) = (n/\sigma)^{\frac{\DELTA}{\DELTA+1}}$.
For each sub-block $I'_\ell$, we create a copy $T_\ell$ of the tree $T(\DELTA, q := (n/\sigma)^{\frac{1}{\DELTA+1}})$ of~\Cref{lem:T_Delta_q},
with mutually disjoint color palettes for the different copies.
Each tree $T_\ell$ has $q^\DELTA = (n/\sigma)^{\frac{\DELTA}{\DELTA+1}}$ leaves; we connect it to $G$ by identifying them with the sub-block $I'_\ell$.
Let $G'$ be the resulting graph after connecting these $\sigma$ trees to the block $I'$.
We've only added $\leq \DELTA n = O(n)$ vertices (and edges) to get $G'$ from $G$.
Note that $G'$ is a partially-colored graph with color classes having size at most $\DELTA$: colors appear only on the trees, and the $G$ part remains uncolored, i.e., non-faulty.
Let $S = \{s_1, \dots, s_\sigma\}$ be the roots of $T_1, \dots, T_\sigma$.
We construct a CFT $S \times V(G')$ distance preserver $H' \subseteq G'$ using the assumed sparsification algorithm.

Next, we iterate over every $i \in I'$, with our goal to find every $k \in K$ such that $i,k$ have a $2$-path.
Let $I'_\ell$ be the sub-block containing $i$, and let $c_i$ be the color corresponding to $i$ as a leaf of $T_\ell$: namely, $i$ is the unique leaf of minimum distance from the root $s_\ell$ in $T_\ell - c_i$.
We run BFS from $s_\ell$ in $H'-c_i$, which gives us all distances from $s_\ell$ to every other vertex in $H'-c_i$, and thus also in $G'-c$ (as $H'$ is a preserver of $G'$).
Finally, for each $k \in K$, we check if the following holds:
\begin{equation}\label{eq:dist-comparison}
    \dist_{G'-c_i}(s_\ell,k) \overset{(?)}{\leq} \dist_{G'-c_i}(s_\ell, i) + 2.
\end{equation}
This achieves our goal, by the following claim:
\begin{claim}
    There is a $2$-path between $i,k$ in $G$ iff~\Cref{eq:dist-comparison} holds.
\end{claim}
\begin{proof}[Proof (of claim).]   
    If there is a $2$-path between $i,k$ in $G$, then it survives in $G'-c_i$ (since the colors in $G'$ only appear on added trees),
    so~\Cref{eq:dist-comparison} holds by triangle inequality.
    
    Conversely, suppose~\Cref{eq:dist-comparison} holds.
    By our construction of $G'$, the shortest path between $s_\ell$ and $k$ in $G'-c_i$ starts by walking from the root $s_\ell$ on the tree $T_\ell - c_i$ to some leaf $i' \in I'_j$.
    So,
    \begin{align}
        \dist_{G'-c_i}(s_\ell, k) 
        &= \dist_{T_\ell - c_i}(s_\ell, i') + \dist_{G'-c_i}(i', k) \nonumber \\
        & \geq \dist_{T_\ell -c_i} (s_\ell, i) + \dist_{G'-c_i}(i', k) && \text{as $i$ has min.\ depth in leaves of $T_\ell -c_i$,} \label{eq:x-has-min-depth}\\
        & \geq \dist_{T_\ell -c_i} (s_\ell, i) + 2 && \text{as there are no edges between $I$ and $K$,} \label{eq:X-Z-dist-is-2-or-more} \\
        & \geq \dist_{G'-c_i}(s_\ell, i) + 2  && \text{since $T_\ell -c_i$ is a subgraph of $G'-c_i$,} \nonumber \\
        & \geq \dist_{G'-c_i}(s_\ell, k) && \text{as~\Cref{eq:dist-comparison} holds.} \nonumber 
    \end{align}
    Hence, all the above inequalities in fact hold with equality.
    By the equality in~\Cref{eq:x-has-min-depth} we have $\dist_{T_\ell-c_i}(s_\ell, i') = \dist_{T_\ell-c_i} (s_\ell, i)$.
    But $i$ is the unique leaf of minimum depth in $T_\ell - c_i$, so we conclude that $i'=i$.
    Now, by the equality in~\Cref{eq:X-Z-dist-is-2-or-more} we have $\dist_{G'-c}(i',k) = 2$, meaning there is a $2$-path between $i'$ (which is $i$) and $k$, as we needed to show.
\end{proof}

We now analyze the algorithm's running time.
Denote the size and running time guarantees of the assumed sparsification algorithm by
\[
s = m^{1-\epsilon} = n^{1+a - \epsilon + o(1)}, \quad t = n^{2 + a -\frac{1-b}{\DELTA+1} - \epsilon} = n^{2 + a - \epsilon + o(1)} / r.
\]
Consider the time spent on block $I'$:
\begin{itemize}[itemsep=0pt]
    \item Constructing $G'$ from $G$ amounts to constructing the trees $T_1, \dots, T_\sigma$.
    Each such tree is constructed in linear time in its size, and the total size of the trees is $O(n)$.
    \item Applying the sparsification to get the preserver $H'$ takes up to $O(t)$ time.
    \item Then, for each $i \in I'$, we run one BFS procedure in (a subgraph of) $H'$ and check $n$ inequalities involving the output distances, which takes $O(s + n)$ time.
\end{itemize}
Thus, the total running time of our BMM algorithm is 
\[
\underbrace{r}_{\text{\#blocks}} \cdot
\left( 
    \underbrace{O(n)}_\text{construct $G'$} +    \underbrace{O(t)}_\text{construct $H'$} +  \underbrace{(n/r)}_{|I'|} \cdot \underbrace{O(s + n)}_\text{BFS in $H'$} 
\right) 
= O(rt + ns + n^2)
= O(n^{2+a-\epsilon + o(1)} + n^2).
\]
Since the sparsification algorithm for constructing preservers is combinatorial, 
our resulting BMM algorithm is also combinatorial, and its running time (which is faster than $\Omega((mn)^{1-o(1)} + n^2)$) contradicts~\Cref{conj:BMM}.
\end{proof}

\begin{remark}\label{remark:data-structure}
    Observe that a CFT $S \times V$ distance preserver $H$ of $G$ immediately provides a \emph{data structure} that, given a query $(s,c)$ of a source $s \in S$ and a failing color $c$, reports the single-source distances in $G-c$, i.e.\ $\{\dist_{G - c} (s,t) \mid t \in V\}$, in $O(|E(H)|)$ time.
    Indeed, the query is answered by computing BFS rooted at $s$ in $H-c$. 
    So, the size guarantee of the algorithm that computes $H$ translates to the \emph{query time} of the data structure, and its running time becomes the \emph{preprocessing time}.

    In the proof of~\Cref{thm:alg-lower-bound}, we actually used this implied data structure and not the preserver directly.
    Thus, the lower bound obtained there also hold for the data structure problem.
    Furthermore, as a $\DELTA$-FT $S \times V$ distance preserver is also a CFT preserver when color classes have size bounded by $\DELTA$, the lower bound also holds for the $\DELTA$-FT version of the data structure, that
    given a query $(s,F)$ with $s \in S$ and $F$ a set of at most $\DELTA$ failing vertices/edges, reports the single-source distances from $s$ in $G-F$.
\end{remark}

\phantomsection
\addcontentsline{toc}{section}{References}
\bibliographystyle{alphaurl}
{\small \bibliography{references}}

\newcommand{\etalchar}[1]{$^{#1}$}
\begin{thebibliography}{CGMW18}

\bibitem[ABCP98]{ABCP98}
Baruch Awerbuch, Bonnie Berger, Lenore Cowen, and David Peleg.
\newblock Near-linear time construction of sparse neighborhood covers.
\newblock {\em SIAM Journal on Computing}, 28(1):263--277, 1998.
\newblock \href {https://doi.org/10.1137/S0097539794271898} {\path{doi:10.1137/S0097539794271898}}.

\bibitem[ABK{\etalchar{+}}02]{AfekBKCM02}
Yehuda Afek, Anat Bremler{-}Barr, Haim Kaplan, Edith Cohen, and Michael Merritt.
\newblock Restoration by path concatenation: fast recovery of {MPLS} paths.
\newblock {\em Distributed Comput.}, 15(4):273--283, 2002.
\newblock \href {https://doi.org/10.1007/S00446-002-0080-6} {\path{doi:10.1007/S00446-002-0080-6}}.

\bibitem[AFK{\etalchar{+}}24]{AbboudFKLM24}
Amir Abboud, Nick Fischer, Zander Kelley, Shachar Lovett, and Raghu Meka.
\newblock New graph decompositions and combinatorial boolean matrix multiplication algorithms.
\newblock In {\em Proceedings of the 56th Annual {ACM} Symposium on Theory of Computing, {STOC} 2024, Vancouver, BC, Canada, June 24-28, 2024}, pages 935--943. {ACM}, 2024.
\newblock \href {https://doi.org/10.1145/3618260.3649696} {\path{doi:10.1145/3618260.3649696}}.

\bibitem[AV14]{AbboudW14}
Amir Abboud and Virginia {Vassilevska Williams}.
\newblock Popular conjectures imply strong lower bounds for dynamic problems.
\newblock In {\em 55th {IEEE} Annual Symposium on Foundations of Computer Science, {FOCS} 2014, Philadelphia, PA, USA, October 18-21, 2014}, pages 434--443. {IEEE} Computer Society, 2014.
\newblock \href {https://doi.org/10.1109/FOCS.2014.53} {\path{doi:10.1109/FOCS.2014.53}}.

\bibitem[BCDW24]{bansal_et_al:LIPIcs.ISAAC.2024.9}
Shivam Bansal, Keerti Choudhary, Harkirat Dhanoa, and Harsh Wardhan.
\newblock {Fault-Tolerant Bounded Flow Preservers}.
\newblock In {\em 35th International Symposium on Algorithms and Computation (ISAAC 2024)}, volume 322 of {\em Leibniz International Proceedings in Informatics (LIPIcs)}, pages 9:1--9:14. Schloss Dagstuhl -- Leibniz-Zentrum f{\"u}r Informatik, 2024.
\newblock \href {https://doi.org/10.4230/LIPIcs.ISAAC.2024.9} {\path{doi:10.4230/LIPIcs.ISAAC.2024.9}}.

\bibitem[BCR18]{BaswanaCR18}
Surender Baswana, Keerti Choudhary, and Liam Roditty.
\newblock Fault-tolerant subgraph for single-source reachability: General and optimal.
\newblock {\em {SIAM} J. Comput.}, 47(1):80--95, 2018.
\newblock \href {https://doi.org/10.1137/16M1087643} {\path{doi:10.1137/16M1087643}}.

\bibitem[BGPW17]{BodwinGPW17}
Greg Bodwin, Fabrizio Grandoni, Merav Parter, and Virginia~Vassilevska Williams.
\newblock Preserving distances in very faulty graphs.
\newblock In {\em 44th International Colloquium on Automata, Languages, and Programming, {ICALP} 2017, July 10-14, 2017, Warsaw, Poland}, volume~80 of {\em LIPIcs}, pages 73:1--73:14. Schloss Dagstuhl - Leibniz-Zentrum f{\"{u}}r Informatik, 2017.
\newblock \href {https://doi.org/10.4230/LIPICS.ICALP.2017.73} {\path{doi:10.4230/LIPICS.ICALP.2017.73}}.

\bibitem[Bod19]{Bodwin19}
Greg Bodwin.
\newblock On the structure of unique shortest paths in graphs.
\newblock In Timothy~M. Chan, editor, {\em Proceedings of the Thirtieth Annual {ACM-SIAM} Symposium on Discrete Algorithms, {SODA} 2019, San Diego, California, USA, January 6-9, 2019}, pages 2071--2089. {SIAM}, 2019.
\newblock \href {https://doi.org/10.1137/1.9781611975482.125} {\path{doi:10.1137/1.9781611975482.125}}.

\bibitem[Bod21]{Bodwin21}
Greg Bodwin.
\newblock New results on linear size distance preservers.
\newblock {\em {SIAM} J. Comput.}, 50(2):662--673, 2021.
\newblock \href {https://doi.org/10.1137/19M123662X} {\path{doi:10.1137/19M123662X}}.

\bibitem[Bod22]{Bodwin22}
Greg Bodwin.
\newblock A note on distance-preserving graph sparsification.
\newblock {\em Inf. Process. Lett.}, 174:106205, 2022.
\newblock \href {https://doi.org/10.1016/J.IPL.2021.106205} {\path{doi:10.1016/J.IPL.2021.106205}}.

\bibitem[BP23]{BodwinP23}
Greg Bodwin and Merav Parter.
\newblock Restorable shortest path tiebreaking for edge-faulty graphs.
\newblock {\em J. {ACM}}, 70(5):28:1--28:24, 2023.
\newblock \href {https://doi.org/10.1145/3603542} {\path{doi:10.1145/3603542}}.

\bibitem[BV15]{BodwinW15}
Greg Bodwin and Virginia {Vassilevska Williams}.
\newblock Very sparse additive spanners and emulators.
\newblock In {\em Proceedings of the 2015 Conference on Innovations in Theoretical Computer Science, {ITCS} 2015, Rehovot, Israel, January 11-13, 2015}, pages 377--382. {ACM}, 2015.
\newblock \href {https://doi.org/10.1145/2688073.2688103} {\path{doi:10.1145/2688073.2688103}}.

\bibitem[BW21]{BodwinW21}
Greg Bodwin and Virginia~Vassilevska Williams.
\newblock Better distance preservers and additive spanners.
\newblock {\em {ACM} Trans. Algorithms}, 17(4):36:1--36:24, 2021.
\newblock \href {https://doi.org/10.1145/3490147} {\path{doi:10.1145/3490147}}.

\bibitem[CC19]{ChechikC19}
Shiri Chechik and Sarel Cohen.
\newblock Near optimal algorithms for the single source replacement paths problem.
\newblock In {\em Proceedings of the Thirtieth Annual {ACM-SIAM} Symposium on Discrete Algorithms, {SODA} 2019, San Diego, California, USA, January 6-9, 2019}, pages 2090--2109. {SIAM}, 2019.
\newblock \href {https://doi.org/10.1137/1.9781611975482.126} {\path{doi:10.1137/1.9781611975482.126}}.

\bibitem[CDP{\etalchar{+}}07]{CoudertDPRV07}
David Coudert, Pallab Datta, Stephane Perennes, Herv{\'{e}} Rivano, and Marie{-}Emilie Voge.
\newblock Shared risk resource group complexity and approximability issues.
\newblock {\em Parallel Process. Lett.}, 17(2):169--184, 2007.
\newblock \href {https://doi.org/10.1142/S0129626407002958} {\path{doi:10.1142/S0129626407002958}}.

\bibitem[CE06]{CoppersmithE06}
Don Coppersmith and Michael Elkin.
\newblock Sparse sourcewise and pairwise distance preservers.
\newblock {\em {SIAM} J. Discret. Math.}, 20(2):463--501, 2006.
\newblock \href {https://doi.org/10.1137/050630696} {\path{doi:10.1137/050630696}}.

\bibitem[CGK13]{CyganGK13}
Marek Cygan, Fabrizio Grandoni, and Telikepalli Kavitha.
\newblock On pairwise spanners.
\newblock In {\em 30th International Symposium on Theoretical Aspects of Computer Science, {STACS} 2013, Kiel, Germany, February 27 - March 2, 2013}, volume~20 of {\em LIPIcs}, pages 209--220. Schloss Dagstuhl - Leibniz-Zentrum f{\"{u}}r Informatik, 2013.
\newblock \href {https://doi.org/10.4230/LIPICS.STACS.2013.209} {\path{doi:10.4230/LIPICS.STACS.2013.209}}.

\bibitem[CGMW18]{ChangGMW18}
Hsien{-}Chih Chang, Pawel Gawrychowski, Shay Mozes, and Oren Weimann.
\newblock Near-optimal distance emulator for planar graphs.
\newblock In {\em 26th Annual European Symposium on Algorithms, {ESA} 2018, Helsinki, Finland, August 20-22, 2018}, volume 112 of {\em LIPIcs}, pages 16:1--16:17. Schloss Dagstuhl - Leibniz-Zentrum f{\"{u}}r Informatik, 2018.
\newblock \href {https://doi.org/10.4230/LIPICS.ESA.2018.16} {\path{doi:10.4230/LIPICS.ESA.2018.16}}.

\bibitem[CLPR09]{ChechikLPR09}
Shiri Chechik, Michael Langberg, David Peleg, and Liam Roditty.
\newblock Fault-tolerant spanners for general graphs.
\newblock In {\em Proceedings of the 41st Annual {ACM} Symposium on Theory of Computing, {STOC}}, pages 435--444. {ACM}, 2009.
\newblock \href {https://doi.org/10.1145/1536414.1536475} {\path{doi:10.1145/1536414.1536475}}.

\bibitem[CM20]{ChechikM20}
Shiri Chechik and Ofer Magen.
\newblock Near optimal algorithm for the directed single source replacement paths problem.
\newblock In {\em 47th International Colloquium on Automata, Languages, and Programming, {ICALP} 2020, Saarbr{\"{u}}cken, Germany (Virtual Conference), July 8-11, 2020}, volume 168 of {\em LIPIcs}, pages 81:1--81:17. Schloss Dagstuhl - Leibniz-Zentrum f{\"{u}}r Informatik, 2020.
\newblock \href {https://doi.org/10.4230/LIPICS.ICALP.2020.81} {\path{doi:10.4230/LIPICS.ICALP.2020.81}}.

\bibitem[CZ24]{ChechikZ24}
Shiri Chechik and Tianyi Zhang.
\newblock Faster algorithms for dual-failure replacement paths.
\newblock In {\em 51st International Colloquium on Automata, Languages, and Programming, {ICALP} 2024, Tallinn, Estonia, July 8-12, 2024}, volume 297 of {\em LIPIcs}, pages 41:1--41:20. Schloss Dagstuhl - Leibniz-Zentrum f{\"{u}}r Informatik, 2024.
\newblock \href {https://doi.org/10.4230/LIPICS.ICALP.2024.41} {\path{doi:10.4230/LIPICS.ICALP.2024.41}}.

\bibitem[DK11]{DinitzK11}
Michael Dinitz and Robert Krauthgamer.
\newblock Fault-tolerant spanners: better and simpler.
\newblock In {\em Proceedings of the 30th Annual {ACM} Symposium on Principles of Distributed Computing, {PODC}}, pages 169--178, 2011.
\newblock \href {https://doi.org/10.1145/1993806.1993830} {\path{doi:10.1145/1993806.1993830}}.

\bibitem[EP16]{ElkinP16}
Michael Elkin and Seth Pettie.
\newblock A linear-size logarithmic stretch path-reporting distance oracle for general graphs.
\newblock {\em {ACM} Trans. Algorithms}, 12(4):50:1--50:31, 2016.
\newblock \href {https://doi.org/10.1145/2888397} {\path{doi:10.1145/2888397}}.

\bibitem[GJM20]{GuptaJM20}
Manoj Gupta, Rahul Jain, and Nitiksha Modi.
\newblock Multiple source replacement path problem.
\newblock In {\em {PODC} '20: {ACM} Symposium on Principles of Distributed Computing, Virtual Event, Italy, August 3-7, 2020}, pages 339--348. {ACM}, 2020.
\newblock \href {https://doi.org/10.1145/3382734.3405714} {\path{doi:10.1145/3382734.3405714}}.

\bibitem[GK17]{GuptaK17}
Manoj Gupta and Shahbaz Khan.
\newblock Multiple source dual fault tolerant {BFS} trees.
\newblock In {\em 44th International Colloquium on Automata, Languages, and Programming, {ICALP} 2017, July 10-14, 2017, Warsaw, Poland}, volume~80 of {\em LIPIcs}, pages 127:1--127:15. Schloss Dagstuhl - Leibniz-Zentrum f{\"{u}}r Informatik, 2017.
\newblock \href {https://doi.org/10.4230/LIPICS.ICALP.2017.127} {\path{doi:10.4230/LIPICS.ICALP.2017.127}}.

\bibitem[GKP17]{GhaffariKP17}
Mohsen Ghaffari, David~R. Karger, and Debmalya Panigrahi.
\newblock Random contractions and sampling for hypergraph and hedge connectivity.
\newblock In {\em Proceedings of the Twenty-Eighth Annual {ACM-SIAM} Symposium on Discrete Algorithms, {SODA}}, pages 1101--1114. {SIAM}, 2017.
\newblock \href {https://doi.org/10.1137/1.9781611974782.71} {\path{doi:10.1137/1.9781611974782.71}}.

\bibitem[GP16]{GhaffariP16}
Mohsen Ghaffari and Merav Parter.
\newblock Near-optimal distributed algorithms for fault-tolerant tree structures.
\newblock In {\em Proceedings of the 28th {ACM} Symposium on Parallelism in Algorithms and Architectures, {SPAA} 2016, Asilomar State Beach/Pacific Grove, CA, USA, July 11-13, 2016}, pages 387--396. {ACM}, 2016.
\newblock \href {https://doi.org/10.1145/2935764.2935795} {\path{doi:10.1145/2935764.2935795}}.

\bibitem[HXX25]{HoppenworthXX25}
Gary Hoppenworth, Yinzhan Xu, and Zixuan Xu.
\newblock New separations and reductions for directed hopsets and preservers.
\newblock In {\em Proceedings of the 2025 Annual {ACM-SIAM} Symposium on Discrete Algorithms, {SODA} 2025, New Orleans, LA, USA, January 12-15, 2025}, pages 4405--4443. {SIAM}, 2025.
\newblock \href {https://doi.org/10.1137/1.9781611978322.150} {\path{doi:10.1137/1.9781611978322.150}}.

\bibitem[Kui12]{Kuipers12}
F.~Kuipers.
\newblock An overview of algorithms for network survivability.
\newblock {\em ISRN Communications and Networking}, 2012, 12 2012.
\newblock \href {https://doi.org/10.5402/2012/932456} {\path{doi:10.5402/2012/932456}}.

\bibitem[LMSZ19]{lokshtanov2019briefnotesinglesource}
Daniel Lokshtanov, Pranabendu Misra, Saket Saurabh, and Meirav Zehavi.
\newblock A brief note on single source fault tolerant reachability, 2019.
\newblock URL: \url{https://arxiv.org/abs/1904.08150}, \href {https://arxiv.org/abs/1904.08150} {\path{arXiv:1904.08150}}.

\bibitem[NS24]{NeimanS24}
Ofer Neiman and Idan Shabat.
\newblock On the size overhead of pairwise spanners.
\newblock In {\em 15th Innovations in Theoretical Computer Science Conference, {ITCS} 2024, Berkeley, CA, USA, January 30 - February 2, 2024}, volume 287 of {\em LIPIcs}, pages 83:1--83:22. Schloss Dagstuhl - Leibniz-Zentrum f{\"{u}}r Informatik, 2024.
\newblock \href {https://doi.org/10.4230/LIPICS.ITCS.2024.83} {\path{doi:10.4230/LIPICS.ITCS.2024.83}}.

\bibitem[Par15]{Parter15}
Merav Parter.
\newblock Dual failure resilient {BFS} structure.
\newblock In {\em Proceedings of the 2015 {ACM} Symposium on Principles of Distributed Computing, {PODC} 2015, Donostia-San Sebasti{\'{a}}n, Spain, July 21 - 23, 2015}, pages 481--490. {ACM}, 2015.
\newblock \href {https://doi.org/10.1145/2767386.2767408} {\path{doi:10.1145/2767386.2767408}}.

\bibitem[Par17]{Parter17}
Merav Parter.
\newblock Vertex fault tolerant additive spanners.
\newblock {\em Distributed Comput.}, 30(5):357--372, 2017.
\newblock URL: \url{https://doi.org/10.1007/s00446-015-0252-9}, \href {https://doi.org/10.1007/S00446-015-0252-9} {\path{doi:10.1007/S00446-015-0252-9}}.

\bibitem[Par20]{Parter20}
Merav Parter.
\newblock Distributed constructions of dual-failure fault-tolerant distance preservers.
\newblock In {\em 34th International Symposium on Distributed Computing, {DISC} 2020, October 12-16, 2020, Virtual Conference}, volume 179 of {\em LIPIcs}, pages 21:1--21:17. Schloss Dagstuhl - Leibniz-Zentrum f{\"{u}}r Informatik, 2020.
\newblock \href {https://doi.org/10.4230/LIPICS.DISC.2020.21} {\path{doi:10.4230/LIPICS.DISC.2020.21}}.

\bibitem[Pet09]{Pettie09}
Seth Pettie.
\newblock Low distortion spanners.
\newblock {\em {ACM} Trans. Algorithms}, 6(1):7:1--7:22, 2009.
\newblock \href {https://doi.org/10.1145/1644015.1644022} {\path{doi:10.1145/1644015.1644022}}.

\bibitem[PP13]{ParterP13}
Merav Parter and David Peleg.
\newblock Sparse fault-tolerant {BFS} trees.
\newblock In {\em Algorithms - {ESA} 2013 - 21st Annual European Symposium, Sophia Antipolis, France, September 2-4, 2013. Proceedings}, volume 8125 of {\em Lecture Notes in Computer Science}, pages 779--790. Springer, 2013.
\newblock \href {https://doi.org/10.1007/978-3-642-40450-4\_66} {\path{doi:10.1007/978-3-642-40450-4\_66}}.

\bibitem[PP15]{ParterP15}
Merav Parter and David Peleg.
\newblock Fault tolerant {BFS} structures: {A} reinforcement-backup tradeoff.
\newblock In {\em Proceedings of the 27th {ACM} on Symposium on Parallelism in Algorithms and Architectures, {SPAA} 2015, Portland, OR, USA, June 13-15, 2015}, pages 264--273. {ACM}, 2015.
\newblock \href {https://doi.org/10.1145/2755573.2755590} {\path{doi:10.1145/2755573.2755590}}.

\bibitem[PP16]{ParterP16}
Merav Parter and David Peleg.
\newblock Sparse fault-tolerant {BFS} structures.
\newblock {\em {ACM} Trans. Algorithms}, 13(1):11:1--11:24, 2016.
\newblock \href {https://doi.org/10.1145/2976741} {\path{doi:10.1145/2976741}}.

\bibitem[PP18]{ParterP18}
Merav Parter and David Peleg.
\newblock Fault-tolerant approximate {BFS} structures.
\newblock {\em {ACM} Trans. Algorithms}, 14(1):10:1--10:15, 2018.
\newblock \href {https://doi.org/10.1145/3022730} {\path{doi:10.1145/3022730}}.

\bibitem[PP20]{ParterP20}
Merav Parter and David Peleg.
\newblock Fault tolerant approximate {BFS} structures with additive stretch.
\newblock {\em Algorithmica}, 82(12):3458--3491, 2020.
\newblock \href {https://doi.org/10.1007/S00453-020-00734-2} {\path{doi:10.1007/S00453-020-00734-2}}.

\bibitem[PPST25]{ParterPST25}
Merav Parter, Asaf Petruschka, Shay Sapir, and Elad Tzalik.
\newblock Parks and recreation: Color fault-tolerant spanners made local.
\newblock In {\em Proceedings of the 2025 Annual {ACM-SIAM} Symposium on Discrete Algorithms, {SODA} 2025, New Orleans, LA, USA, January 12-15, 2025}, pages 4061--4094. {SIAM}, 2025.
\newblock \href {https://doi.org/10.1137/1.9781611978322.139} {\path{doi:10.1137/1.9781611978322.139}}.

\bibitem[PST24]{PetruschkaST24}
Asaf Petruschka, Shay Sapir, and Elad Tzalik.
\newblock Connectivity labeling in faulty colored graphs.
\newblock In {\em 38th International Symposium on Distributed Computing, {DISC} 2024, Madrid, Spain, October 28 - November 1, 2024}, volume 319 of {\em LIPIcs}, pages 36:1--36:22. Schloss Dagstuhl - Leibniz-Zentrum f{\"{u}}r Informatik, 2024.
\newblock \href {https://doi.org/10.4230/LIPICS.DISC.2024.36} {\path{doi:10.4230/LIPICS.DISC.2024.36}}.

\bibitem[PST26]{PetruschkaST26}
Asaf Petruschka, Shay Sapir, and Elad Tzalik.
\newblock Color fault-tolerant spanners.
\newblock {\em {ACM} Trans. Algorithms}, 22(1):6:1--6:21, 2026.
\newblock \href {https://doi.org/10.1145/3750728} {\path{doi:10.1145/3750728}}.

\bibitem[RZ11]{RodittyZ11}
Liam Roditty and Uri Zwick.
\newblock On dynamic shortest paths problems.
\newblock {\em Algorithmica}, 61(2):389--401, 2011.
\newblock \href {https://doi.org/10.1007/S00453-010-9401-5} {\path{doi:10.1007/S00453-010-9401-5}}.

\bibitem[VW18]{WilliamsW18}
Virginia {Vassilevska Williams} and R.~Ryan Williams.
\newblock Subcubic equivalences between path, matrix, and triangle problems.
\newblock {\em J. {ACM}}, 65(5):27:1--27:38, 2018.
\newblock \href {https://doi.org/10.1145/3186893} {\path{doi:10.1145/3186893}}.

\bibitem[VWX22]{WilliamsWX22}
Virginia {Vassilevska Williams}, Eyob Woldeghebriel, and Yinzhan Xu.
\newblock Algorithms and lower bounds for replacement paths under multiple edge failure.
\newblock In {\em 63rd {IEEE} Annual Symposium on Foundations of Computer Science, {FOCS} 2022, Denver, CO, USA, October 31 - November 3, 2022}, pages 907--918. {IEEE}, 2022.
\newblock \href {https://doi.org/10.1109/FOCS54457.2022.00090} {\path{doi:10.1109/FOCS54457.2022.00090}}.

\bibitem[ZCTZ11]{ZPT11}
Peng Zhang, Jin-Yi Cai, Lin-Qing Tang, and Wen-Bo Zhao.
\newblock Approximation and hardness results for label cut and related problems.
\newblock {\em Journal of Combinatorial Optimization}, 21:192--208, 2011.
\newblock \href {https://doi.org/10.1007/s10878-009-9222-0} {\path{doi:10.1007/s10878-009-9222-0}}.

\end{thebibliography}

\appendix

\section{Other Types of CFT Preservers}\label{sec:other-CFT}

\subsection{Applications of Multi-Source: CFT Pairwise Preservers and $+2$ Spanners}\label{app:applications}

The first application is for \emph{pairwise} CFT distance preservers with arbitrary demand pairs:

\begin{theorem}
    For any $n$-vertex colored graph $G = (V,E)$ with color classes of size at most $\DELTA$, and any $P \subseteq V \times V$, there exists a CFT $P$-distance preserver $H$ with 
    $
    \tilde{O}(n^{2-\frac{2}{2\DELTA+3}} \cdot |P|^{\frac{1}{2\DELTA+3}})
    $
    edges.
\end{theorem}

\begin{proof}
    Let $L$ be a parameter to be set later;
    a path is \emph{long} if its length is $\geq L$, and \emph{short} otherwise.
    The strategy is to separately take care of triplets $(s,t,c)$ where $(s,t) \in P$ and $c$ is some color, depending on whether $\pi(s,t \mid c)$ is long (``long triplet'') or short (``short triplet'').
    \begin{itemize}
        \item \emph{Long triplets:}
        These are handled using CFT multi-source preservers.
        Let $S \subseteq V$ be a hitting set for all $\pi(s,t \mid c)$ with $(s,t,c)$ a long triplet.
        By standard hitting set arguments, $|S| = O(\frac{n}{L} \cdot \log n)$.
        We add to $H$ the CFT $S \times V$ distance preservers of~\Cref{thm:size-upper-bound}, which has $\tilde{O}(n^2 \cdot L^{-\frac{1}{\DELTA+1}})$ edges.
        Now, if $(s,t,c)$ is a long triplet, there is some $x \in S \cap \pi(s,t \mid c)$, so
        \begin{align*}
        \dist_{H-c} (s,t)
        &\leq \dist_{H-c} (s,x) + \dist_{H-c} (x,t) && \text{by triangle inequality,} \\
        &= \dist_{G-c} (s,x) + \dist_{G-c} (x,t) && \text{as $H$ is a CFT $(S \times V) \cup (V \times S)$ dist.\ pres.,} \\
        &= \dist_{G-c} (s,t) &&\text{since $x \in \pi(s,t \mid c)$}.
        \end{align*}

        \item \emph{Short triplets:}
        These are handled using ``hit-miss sampling''.
        We work in $r = 40 L \ln n$ iterations.
        In iteration $i$, we sample each color independently with probability $\frac{1}{L}$ into a color set $F_i$, and add to $H$ every \emph{short} path of the form $\pi(s,t \mid F_i)$ with $(s,t) \in P$.
        This adds $\leq |P| \cdot L$ edges in each iteration, hence $\tilde{O}(|P| \cdot L^2)$ edges over all iterations.
        
        Now, suppose $\pi(s,t \mid c)$ is short.
        We say iteration $i$ is \emph{good} for the triplet $(s,t,c)$ if $c \in F_i$, and no color from $\pi(s,t\mid c)$ is in $F_i$.
        In this case, iteration $i$ adds to $H$ a path of length $\dist_{G-c} (s,t)$ which avoids the color $c$, so $(s,t,c)$ is taken care of.
        The probability of a single iteration to be good is $\frac{1}{L}(1-\frac{1}{L})^L \geq \frac{1}{4L}$.
        So, the probability that $(s,t,c)$ has some good iteration is at least $1 - (1-\frac{1}{4L} )^r \geq 1 - e^{-\frac{r}{4L}} \geq 1 - n^{-10}$. By a union bound, with high probability every short triplet $(s,t,c)$ has some good iteration.
    \end{itemize}

   Choosing $L$ s.t.\
   $|P| \cdot L^2 = n^2 \cdot L^{-\frac{1}{\DELTA+1}}$, i.e., 
   $L = (\frac{n^2}{|P|})^{\frac{\DELTA+1}{2\DELTA+3}}$,
   yields the stated size bound for $H$.
\end{proof}

The second application is for CFT $+2$ \emph{spanners}: this is a subgraph $H$ of $G$ such that $\dist_{H-c} (u,v) \leq \dist_{G-c} (u,v) + 2$ for every $u,v \in V$ and color $c$.
The proof idea is based on~\cite{BodwinGPW17}.

\begin{theorem}
    For any $n$-vertex colored graph $G = (V,E)$ with color classes of size at most $\DELTA$, there exists a CFT $+2$ spanner $H$ with 
    $
    |E(H)| = \tilde{O} (n^{2-\frac{1}{\DELTA+2}})
    $
    edges.
\end{theorem}

\begin{proof}
    We show the proof for the case where $G$ has colored edges, and the case of colored vertices is very similar.
    Let $L$ be a parameter to be set later.
    We say a vertex is \emph{colorful} if it has at least $L$ edges of different colors incident to it, and \emph{dull} otherwise.
    
    Let $S \subseteq V$ be such that every colorful vertex has at least two incident edges of different colors with the other endpoint in $S$.
    By standard hitting set arguments, $|S| = O(\frac{n}{L} \cdot \log n)$. 
    Let $H_{dull} \subseteq G$ consists of all the edges incident to dull vertices.
    Let $H_{src}$ be a CFT $S \times V$ distance preserver of~\Cref{thm:size-upper-bound}.
    We define $H = H_{dull} \cup H_{src}$.

    We first prove correctness.
    Fix $u,v \in V$ and color $c$, and we show that $\dist_{H-c} (u,v) \leq \dist_{G-c} (u,v) + 2$.
    If $\pi(u,v \mid c)$ is contained in $H_{dull}$, we are done.
    Otherwise, there is some colorful $x \in \pi(u,v \mid c)$, with at least two incident edges of different colors connecting $x$ to $S$.
    At least one of those is not of color $c$.
    Denote it by $e = (x,s)$ with $s \in S$.
    Thus,
    \begin{align*}
        &\dist_{H-c} (u,v)\\
        &\leq \dist_{H-c} (u,s) + \dist_{H-c}  (s,v) &&\text{by triangle ineq.,} \\
        &= \dist_{G-c} (u,s) + \dist_{G-c} (s,v) &&\text{as $H$ is a  $1$-CFT $S \times V$ preserver,} \\
        &\leq \dist_{G-c} (u,x) + \dist_{G-c} (x,v) + 2 \dist_{G-c} (x,s) &&\text{by triangle ineq.\ (splitting both distances at $x$),}\\
        &= \dist_{G-c} (u,v) + 2 &&\text{since $x \in \pi(u,v \mid c)$, and $e=(x,s)$ is in $G-c$.}
    \end{align*}

    We now analyze the size.
    Note that a dull vertex sees at most $L$ different colors on its incident edges, so its degree can be at most $L \DELTA$.
    Thus, $H_{dull}$ contains at most $nL \DELTA$ edges.
    As for $H_{src}$, by~\Cref{thm:size-upper-bound} it has $O(n^2 \cdot L^{-\frac{1}{\DELTA+1}} \cdot \DELTA \log n)$ edges.
    We set $L$ so that $nL \DELTA = n^2 \cdot L^{-\frac{1}{\DELTA+1}} \cdot \DELTA \log n$, namely $L = (n \log n)^{1-\frac{1}{\DELTA+2}}$.
    Plugging this choice of $L$ yields the stated size bound for $H$.
\end{proof}

\subsection{Single-Pair Distance Preservers in Weighted Graphs}\label{sec:single-pair}

In this section, we study CFT \emph{single-pair} preservers, required to preserve the distance between a specific pair of vertices $s,t$ under any color fault.
Here, we focus on \emph{weighted} graphs with \emph{failing edges} (i.e., edge-colored in the CFT setting), and show a notable gap between $\DELTA$-FT and CFT with color classes of size $\leq \DELTA$, already for $\DELTA=2$.
The $\DELTA$-FT case was studied by~\cite{BodwinGPW17}, who showed that $1$-FT single pair distance preservers only need $O(n)$ edges, while $\DELTA$-FT for any $\DELTA \geq 2$ may require $\Omega(n^2)$ edges.

Our bounds for the CFT case are obtained by reductions to (fault-free) pairwise distance preservers.
Let $\DP(n,p)$ denote the worst-case size of (fault-free) distance preservers for at most $p$ pairs in $n$-vertex weighted undirected graphs.
Coppersmith and Elkin~\cite{CoppersmithE06} proved that $\Omega(n^{2/3}p^{2/3}) \leq \DP(n,p) \leq O(\min\{\sqrt{n}p+n,n\sqrt{p}\})$.
We show:

\begin{theorem}\label{thm:single-pair-informal}
    Let $\CFTSP(n,\DELTA)$ be the worst-case size of a CFT single-pair distance preserver, over $n$-vertex undirected weighted edge-colored graphs 
    s.t.\ each color appears at most $\DELTA$ times.
    Then,
    \[
    \Omega(\DP(n, n)) \leq \CFTSP(n,\DELTA) \leq O(\DP(n, \DELTA n)).
    \]
\end{theorem}
\noindent
When $\DELTA = \poly \log n$, we get $\CFTSP(n,\DELTA) = \widetilde{\Theta}(\DP(n,n))$; 
plugging the bounds of~\cite{CoppersmithE06}, this is between $\tilde{\Omega}(n^{4/3})$ and $\tilde{O}(n^{3/2})$.

We note that unlike other preservers in the paper, in this section we only consider failing edges rather than vertices; this is because our results use the well-known \emph{restoration lemma}~\cite{AfekBKCM02} which only holds for edge failures. 

\begin{lemma}[Weighted Restoration Lemma~\cite{AfekBKCM02}]\label{thm:wtd-restoration}
    After $r$ edge failures in an undirected weighted graph, each new shortest path is a concatenation interleaving at most $r + 1$ original shortest paths and $r$ edges.
\end{lemma}

We now show the upper bound in~\Cref{thm:single-pair-informal}.

\begin{proof}[Proof (of upper bound in~\Cref{thm:single-pair-informal})]
    As noted in~\cite{AfekBKCM02},
    by applying small perturbations to the edge weights, we may assume that all shortest paths in $G$ are unique. 
    Consider a color $c$ appearing on $\pi(s,t)$.
    By the weighted restoration lemma (\Cref{thm:wtd-restoration}),
    \begin{equation}\label{eq:restoration-decomp}
        \pi(s,t \mid c) = \pi(v_0, u_0) \circ e_1 \circ \pi(v_1, u_1) \circ e_2 \circ \cdots \circ e_\ell \circ \pi(v_\ell, u_\ell)
    \end{equation}
    where $\ell \leq \DELTA$, $v_0 = s$, $u_\ell = t$, and each edge $e_i$ connects $v_{i-1}$ and $u_i$.
    Denote the set of edges appearing in this decomposition as $E(c) = \{e_i : 1\leq i \leq \ell\}$ and, the set of intermediate vertex-pairs for which shortest path appear by $P(c) = \{(v_i, u_i) : 1 \leq i \leq \ell-1\}$.
    The subgraph $H$ is constructed as the union of:
    \begin{enumerate}[itemsep=0pt]
        \item Two shortest-paths-trees $T_s$ and $T_t$ in $G$, rooted at $s$ and $t$ (at most $2(n-1)$ edges),
        \item All edges found in some $E(c)$ for $c$ appearing on $\pi_G(s,t)$ (at most $n\DELTA$ edges), and
        \item A (non-FT) pairwise distance preserver $H'$ of $G$ for the pairs $\bigcup_c P(c)$ for every color $c$ on $\pi(s,t)$ (at most $\DP(n, \DELTA n)$ edges).
    \end{enumerate}
    Note that $\DP(n,\DELTA n) = \Omega(\DELTA n)$, as seen by considering an unweighted clique on $n$ vertices with any arbitrary $\DELTA n$ demand pairs.
    So overall, $|E(H)| \leq O(\DP(n,\DELTA n))$.
    
    Finally, we claim that $\pi(s,t \mid c) \subseteq H$ for every color $c$.
    If $c$ doesn't appear on $\pi(s,t)$, then $\pi(s,t \mid c) = \pi(s,t) \subseteq T_s \subseteq H$.
    If $c$ does appear on $\pi(s,t)$, then all of the components in the decomposition of $\pi(s,t \mid c)$ in \Cref{eq:restoration-decomp} are included in $H$: Indeed, the first and last paths are in $T_s \cup T_t$, the edges $e_1, \dots, e_\ell \in E(c)$, and the paths $\pi(v_i, u_i)$ for $1\leq i \leq \ell-1$ are in $H'$. 
\end{proof}

\begin{remark}
    The correctness argument in the above proof crucially uses that shortest paths in $G$ are \emph{unique}, which we obtained by adding small perturbations to edge weights.
    For this reason, even if the original graph $G$ was unweighted, 
    the function $\DP$ is defined in terms of weighted graphs.
    Better bounds are known for the unweighted analog of $\DP$ (Bodwin and Vassilevska Williams~\cite{BodwinW21});
    but these are obtained using sophisticated shortest-paths tiebreaking for unweighted shortest paths which might be ``in conflict'' with the restoration lemma, so we cannot apply them.
\end{remark}

We finish this section by showing the lower bound:

\begin{proof}[Proof (of lower bound in~\Cref{thm:single-pair-informal}]
    Let $G^*$ and $P^*$ be the worst-case instance for $\DP(n,n)$.
    Namely, $G^*$ is a weighted (uncolored) graph with $n$ vertices and $\DP(n,n)$ edges, $P^* = \{(x_i,y_i)\}_{i=1}^n$ is a set of $n$ vertex-pairs from $V(G^*)$, and the following property holds: For every $e \in E(G^*)$, there exists a pair $(x_i, y_i) \in P$ such that $\dist_{G^*}(x_i,y_i) < \dist_{G^* - e}(x_i,y_i)$, i.e., $e$ lies on every shortest $x_i \rsquig y_i$ path in $G^*$.
    We assume that the distance between any two vertices in $G^*$ is strictly smaller than $1$ (otherwise we can rescale the weights).

    We construct the lower bound instance $G$ with pair $s,t$ as follows.
    Start with $G^*$ (which is uncolored).
    Create  new vertices $u_1, \dots, u_{n+1}$ and $v_1, \dots, v_{n+1}$.
    For every $1 \leq i \leq n$, add the following edges:
    \begin{itemize}[itemsep=0pt]
        \item The two edges $(u_i, x_i)$ and $(y_i, v_i)$, both having weight $n-i$ and no color.
        \item The two edge $(u_i, u_{i+1})$ and $(v_i, v_{i+1})$, both having weight $0$ (or very small $\epsilon > 0$) and the same color $c_i$. 
    \end{itemize}
    Finally, let $s = u_1$ and $t = v_1$.
    See Figure~\ref{fig:single-pair-lb} for an illustration.
    \begin{figure}[]
        \centering
        \resizebox{12cm}{!}{\input{single-pair-lb.tikz}} 
        \caption{
            Illustration of the lower bound graph $G$ for $1$-CFT $(s,t)$-distance preserver from~\Cref{thm:single-pair-informal}.
            The graph $G^*$ with pairs $P^* = \{(x_i, y_i)\}_{i=1}^n$ form a worst-case instance for $\DP(n,n)$, reweighted so that all distances in $G^*$ are smaller than $1$.
        }
        \label{fig:single-pair-lb}
    \end{figure}

    Now, suppose $H$ is a CFT $\{(s,t)\}$-distance preserver of $G$.
    We show that each edge $e$ of $G^*$ must also be present $H$.
    Let $(x_i, y_i) \in P^*$ be a pair such that $\dist_{G^*}(x_i,y_i) < \dist_{G^* - e}(x_i,y_i)$.
    The construction of $G$ ensures that when the color $c_i$ fails, any simple $s \rsquig v$ path must have
    \begin{itemize}[itemsep=0pt]
    \item a prefix of the form $(s = u_1, u_2, \dots u_j, x_j)$ for $1\leq j \leq i$,
    \item a suffix of the form $(y_r, v_r, v_{r-1}, \dots, v_1 = t)$ for $1 \leq r \leq i$, and
    \item a connecting ``middle part'' which is an $x_j \rsquig y_r$ path in $G^*$.
    \end{itemize}
    If $j \neq i$ or $r \neq i$, then such a path must have weight at least $(n-j) + (n-r) \geq 2(n-i) + 1$.
    Conversely, if we take $j = r = i$ and use the shortest $x_i \rsquig y_i$ path in $G^*$, the path has weight $2(n-i) + \dist_{G^*}(x_i,y_i) < 2(n-i) + 1$.
    Thus, every shortest $s \rsquig t$ path in $G-c_i$ contains a shortest $x_i \rsquig y_i$ path in $G^*$, and therefore must contain $e$.
    As $H$ is a $1$-CFT $(s,t)$-distance preserver, it has to contain at least one shortest $s \rsquig t$ path from $G-c_i$, so $e$ must be present in $H$.
\end{proof}

\subsection{Single Source Reachability and Bounded Flow in Directed Graphs}

Single-source $\DELTA$-FT reachability preservers were introduced by Baswana, Choudhary and Roditty~\cite{BaswanaCR18}.
Here, the given graph $G = (V,E)$ is \emph{directed} and comes with a single source vertex $s$, and the goal is to find a sparse subgraph $H = (V,E')$ such that for every $t \in V$ and every $F \subseteq V$ ($F \subseteq E$) with $|F| \leq \DELTA$, there is a directed $s \rsquig t$ path in $G-F$ iff there is such a path in $H-F$.
The extension of this definition to the CFT setting is obvious (replacing the failure of the set $F$ with the failure of a color $c$ in the above definition).
The work of~\cite{BaswanaCR18} showed that any $n$-vertex \emph{directed} graph has a $\DELTA$-FT single-source reachability preserver with $O(2^\DELTA \cdot n)$ edges, and gave a matching lower bound.
We strengthen their lower bound and show it holds also in the CFT setting with color classes of size $\leq \DELTA$:

\begin{theorem}\label{thm:single-source-reachability-informal}
    The worst-case size of CFT single-source reachability preservers for $n$-vertex directed colored graphs with color classes of size $\leq \DELTA$ 
    (s.t. $2^{\DELTA} \leq n$) is
    $\widetilde{\Theta} \big (2^{\DELTA} \cdot n)$.
\end{theorem}
Hence, here there is essentially no gap between $\DELTA$-FT and CFT with color classes of size $\leq \DELTA$, i.e., the brute-force approach of using a solution for the $\DELTA$-FT to also solve CFT is optimal.

The problem has been further generalized to $\DELTA$-FT single-source \emph{$\lambda$-bounded flow preservers}, where $H-F$ is not only required to preserve the reachability from the source $s$ in $G-F$, but the \emph{flow value} from $s$ to every other vertex up to a threshold of $\lambda$~\cite{lokshtanov2019briefnotesinglesource,bansal_et_al:LIPIcs.ISAAC.2024.9} 
(the reachability case is just $\lambda=1$).
Bansal, Choudhary, Dhanoa and Wardhan~\cite{bansal_et_al:LIPIcs.ISAAC.2024.9} have recently generalized the results of~\cite{BaswanaCR18} and proved that the optimal (worst-case) sparsity in $\DELTA$-EFT is $\Theta(\lambda 2^{\DELTA} \cdot n)$.
Our CFT lower bound easily generalizes as well, giving another instance of where the CFT setting does not provide a benefit over the corresponding $\DELTA$-FT setting:
\begin{theorem}\label{thm:single-source-flow-informal}
    The worst-case size of CFT single-source $\lambda$-bounded flow preservers for $n$-vertex directed unweighted edge-colored graphs with color classes of size $\leq \DELTA$ 
    (s.t. $\lambda 2^\DELTA \leq n$) is
    $\widetilde{\Theta} \big (\lambda 2^\DELTA \cdot n)$.
\end{theorem}

From now on, we focus on the proof of~\Cref{thm:single-source-flow-informal} (as~\Cref{thm:single-source-reachability-informal} is a special case).
We will arbitrarily consider the case of edge failures (i.e., edge colored graphs), but the extension for vertex failures is trivial.
The upper bound of $O(\lambda 2^{\DELTA} \cdot n)$ follows trivially from the $\DELTA$-FT upper bound of~\cite{bansal_et_al:LIPIcs.ISAAC.2024.9}.
The lower bound is based on the following simple tree construction:

\begin{lemma}\label{lem:unique-leaf-tree}
    For every integer $\DELTA \geq 0$,
    there is a rooted edge-colored tree $T = T(\DELTA)$ such that:
    \begin{enumerate}[label=(\alph*), itemsep=0pt]
        \item For every leaf $v$ of $T$, there is a color $c_v$ such that all leaves but $v$ become disconnected from the root in $T-c_v$.%
        \label{prop:leave-to-color-map-in-T(Delta)}

        \item Each color appears $\leq \DELTA$ edges of $T$.%
        \label{prop:color-classes-in-T(Delta)}
        
        \item $T$ has $2^\DELTA$ leaves.%
        \label{prop:num-leaves-in-T(Delta)}
        
        \item $T$ has $\DELTA 2^\DELTA$ edges.%
        \label{prop:num-edges-in-T(Delta)}
    \end{enumerate}
\end{lemma}

\begin{proof}
    By induction on $\DELTA$.
    For the base case $\DELTA = 0$ we take single vertex $v$, which is both the root and the leaf, and an associated color $c_v$ (that does not appear on $T(0)$ as there are no edges).

    For the induction step, we construct $T = T(\DELTA)$ from $T(\DELTA-1)$ as follows.
    Let $T_0$ and $T_1$ be two disjoint copies of $T(\DELTA-1)$, which are also disjoint in their associated color palettes.
    Let $r_0, r_1$ be the roots of $T_0, T_1$, respectively.
    We create a root $r$ for $T(\DELTA)$, and connect it by a two paths $P_0, P_1$ to $r_0, r_1$ respectively, such that $P_0, P_1$ are vertex-disjoint paths (except in $r$), each having $2^{\DELTA-1}$ edges. 
    For every $i \in \{0,1\}$ and every leaf $v$ in $T_i$, we assign the color $c_v$ to one of the edges in $P_{1-i}$.
    Note that, by the induction hypothesis, there are $2^{\DELTA-1}$ leaves in $T_i$, which is the number of edges in $P_{1-i}$, so the latter can indeed accommodate the colors.
    See Figure~\ref{fig:T(Delta)} for an illustration.
    \begin{figure}[H]
        \centering
        \resizebox{8cm}{!}{\input{binary-tree.tikz}} 
        \caption{
            Construction of the tree $T(\DELTA=3)$.
            $T_0$ and $T_1$ are copies of $T(2)$ with disjoint colors.
            The corresponding color $c_v$ for each leaf $v$ is shown beneath it.
        }
        \label{fig:T(Delta)}
    \end{figure}

    We show Property~\ref{prop:leave-to-color-map-in-T(Delta)}.
    Let $v$ be a leaf of $T$, found in some $T_i$.
    All leaves of $T$ that are in $T_{1-i}$ become disconnected from $r$ in $T-c_v$, because they are all descendants of the edge with color $c_v$ that appears in $P_{1-i}$.
    Every leaf $v' \neq v$ that is also in $T_i$ is disconnected from $r_i$ in $T_i - c_v$ (by the induction hypothesis), hence $v'$ is also disconnected from $r$ in $T - c_v$.
    The leaf $v$ itself remains connected to $r_i$ in $T_i - c_v$ (by induction hypothesis).
    Since $P_i$ does not contain the color $c_v$ (it only contains colors from the other tree $T_{1-i}$), we see that $v$ remains connected to $r$ in $T-c_v$.

    The remaining properties \ref{prop:color-classes-in-T(Delta)}, \ref{prop:num-leaves-in-T(Delta)} and \ref{prop:num-edges-in-T(Delta)} follow immediately from the construction and induction hypothesis. 
\end{proof}

\begin{proof}%
    We first create $\lambda$ copies of $T(\DELTA)$, denoted $T_1, \dots, T_\lambda$, with disjoint color palettes, 
    where the edges of each copy $T_i$ are directed away from its root $r_i$.
    Next, we create a vertex $s$ and add (uncolored) outgoing edges from $s$ into every root $r_i$ of a tree $T_i$.
    The vertex $s$ will be our designated source in our lower bound instance $G$.
    Finally, denote by $X$ the set of leaves in $T_1, \dots, T_\lambda$.
    We create another $n$ vertices, denoted $Y$, and add all (uncolored) directed edges $(x,y) \in X \times Y$ (from $x$ to $y$) to obtain $G$.

    The number of vertices in each $T_i$ is $O(\DELTA \cdot 2^{\DELTA}) = O(n/ \lambda)$ (by \Cref{lem:unique-leaf-tree}\ref{prop:num-edges-in-T(Delta)} and our assumption that $\lambda 2^{\DELTA} \leq n$), hence $G$ has $\lambda \cdot O(n/\lambda) + n = \Theta(n)$ vertices.
    The number of leaves in each $T_i$ is $2^{\DELTA}$ (by \Cref{lem:unique-leaf-tree}\ref{prop:num-leaves-in-T(Delta)}), therefore $|X \times Y| = \lambda 2^{\DELTA} n$.
    So, to conclude the proof, it suffices to show that any CFT $\{s\}\times V$ $\lambda$-bounded flow preserver $H$ must contain every edge $(x,y) \in X\times Y$.

    Let $c_x$ be the color corresponding to the leaf $x$ in the copy $T_i$ which contains $x$, from~\Cref{lem:unique-leaf-tree}\ref{prop:leave-to-color-map-in-T(Delta)}.
    Then $x$ is the only leaf that remains reachable from $r_i$ in $T_i - c_x$.
    Hence, by our construction of $G$, there is a unique $s \rsquig y$ path $P_i$ in $G - c_x$ that starts with the edge $(s,r_i)$ -- this $P_i$ must go through $x$ and use the edge $(x,y)$.
    For each $j \neq i$, we can choose some $s \rsquig y$ path $P_j$ in $G-F_x$ that starts with the edge $(s,r_j)$ and goes down the (non-faulty) $T_j$ to some leaf, and takes one extra edge into $y$.
    Note that $P_1, \dots, P_\lambda$ are edge-disjoint.
    Hence, as $H$ is a CFT single-source $\lambda$-bounded flow preserver, $H-c_x$ must contain $\lambda$ edge-disjoint $s \rsquig y$ paths. 
    As $s$ has exactly $\lambda$ outgoing edges, going into $r_1, \dots, r_\lambda$, one of these paths must start with the edge $(s,r_i)$, so it must be $P_i$.
    Namely, $H$ contains the path $P_i$, and particularly the edge $(x,y)$.
\end{proof}

\section{Multi-Source CFT Distance Preservers in Directed Graphs}\label{sec:multi-source-directed}

In this section we consider multi-source CFT distance preservers for directed graphs: we show that the size bound of~\Cref{thm:size-upper-bound} also holds when $G$ is directed.
We assume that $\DELTA = O(\log n)$ as the size bound is only meaningful in this regime.
We arbitrarily give the proof when $G$ has colored edges, but the proof with colored vertices is virtually identical.

We use the same distance thresholds $d_i$ and hitting sets $A_i$ with~\Cref{prop:hitting-set} from~\Cref{sec:framework}.
We construct the preserver $H$ as follows:
For every $t \in V$ and $0 \leq i \leq \DELTA$, 
\begin{enumerate}[label=\textbf{(R\arabic*)}, itemsep=0pt]
    \item For each $a_i \in A_i$, add $\LastE(a_i, t)$ to $H$.
    \label{rule1}

    \item For each $a_i \in A_i$ and each color $c$ on the $d_{i+1}$-suffix of $\pi(a_i, t)$, add $\LastE(a_i,t \mid c)$ to $H$.
    \label{rule2}
\end{enumerate}

\paragraph{Size Analysis.}
Fix $t$ and $i$.
Rule~\ref{rule1} adds at most $|A_i|$ edges, and Rule~\ref{rule2} adds at most $|A_i| \cdot d_{i+1}$ edges, so overall $O(|A_i| \cdot d_{i+1})$ edges are added for $t,i$.
We note that
\[
|A_i| \cdot d_{i+1}
=
\begin{cases}
O \big( n \cdot \frac{d_{i+1}}{d_i} \cdot \log n \big) & \text{if $i>0$} \\
|S| \cdot d_1 & \text{if $i = 0$}
\end{cases}
~ =
\tilde{O} ( n^{1-\frac{1}{\DELTA+1}} \cdot |S|^{\frac{1}{\DELTA+1}} ).
\]
Summing over all $t \in V$ and $0\leq i \leq \DELTA$ yields the desired bound $|E(H)| = \tilde{O}(n^{2-\frac{1}{\DELTA+1}} \cdot |S|^{\frac{1}{\DELTA+1}})$.

\paragraph{Correctness.}
As in the undirected case, the correctness relies on the last edge observation (\Cref{obs:last-edge}, whose proof also holds when $G$ is directed).
We fix $s\in S$, $t\in V$ and color $c$, and our goal is just to prove that $\LastE(s,t \mid c) \in H$.
To this end, we show an inductive construction of ``(vertex, edge-set)'' pairs
$
(a_0, E_0),
(a_1, E_1),
\dots,
(a_i, E_i),
\dots
$
with the following invariants:
\begin{itemize}
\item[]
\begin{enumerate}[label=\textbf{(I\arabic*)}]
    \item The vertex $a_i \in A_i$ appears on $\pi(s,t \mid c)$ and $\dist_{G-c} (a_i ,t) \leq d_i$.
    \label{inv:a_i}
    
    \item The edge-set $E_i$ contains $i$ edges of color $c$, such that $\dist_G (e, t) > d_i$ for all $e \in E_i$.
    \label{inv:E_i}
\end{enumerate}
\end{itemize}
Intuitively, these invariants imply that at the $i$-th inductive step of the construction, the ``effective'' number of failing edges caused by the failure of the color $c$ is only $\DELTA-i$ or less, as there are $i$ edges of color $c$ that are ``far away'' from $t$, and we can replace $s$ with a ``close by'' vertex $a_i \in A_i$.

We initialize $a_0 = s$ and $E_0 = \emptyset$, so the invariants trivially hold (as $A_0 = S$, $d_0 = n$).
For the inductive step, assuming we already found $(a_i, E_i)$ for some $0 \leq i \leq \DELTA-1$, we either show how to find $(a_{i+1}, E_{i+1})$, or directly prove that $\LastE(s,t \mid c) \in H$.
We first present the cases where the latter option occurs.
By Invariant~\ref{inv:a_i}, we have $\LastE(s,t \mid c)= \LastE(a_i,t \mid c)$, so it suffices for us to show that the latter is in $H$.
\begin{enumerate}[label=(\roman*)]
    \item If color $c$ does not appear on $\pi(a_i, t)$, then 
    $
    \LastE(a_i,t \mid c) = 
    \LastE(a_i,t) \in H
    $
    by Rule~\ref{rule1}.
    \label{case1}

    \item If color $c$ appears on the $d_{i+1}$-suffix of $\pi(a_i, t)$, then $\LastE(a_i,t \mid c) \in H$ by Rule~\ref{rule2}.
    \label{case2}
\end{enumerate}
Suppose now that none of the above cases occur; we show how to find $(a_{i+1}, E_{i+1})$. See illustration in~\Cref{fig:color-sourcewise}.
\begin{figure}
    \centering
    \resizebox{\linewidth}{!}{\input{color-sourcewise-new.tikz}} 
    \caption{
        Illustration of the third inductive step in the correctness proof of~\Cref{thm:size-upper-bound} for directed graphs.
        The red edges correspond to (faulty) edges of color $c$.
        Distance thresholds from $t$ on the illustrated paths are shown as dashed lines.
        \emph{Left:} Case~\ref{case1}, where $\pi(a_2, t)$ is a suffix of $\pi(s,t\mid c)$.
        \emph{Middle:} Case~\ref{case2}, where an edge of color $c$ appears on the $d_3$-length suffix of $\pi(a_2, t)$.
        \emph{Right:} Neither of the cases \ref{case1} and \ref{case2} occurs.
        Then, we find a new $c$-colored edge on $\pi(a_2, t)$ of distance $> d_3$ from $t$, add it to $E_2$ to create $E_3$, and deduce that $\pi(s,t \mid c)$ must contain some $a_3 \in A_3$ on its $d_3$-length suffix.
    }
    \label{fig:color-sourcewise}
\end{figure}
As case~\ref{case1} did not occur, there is some edge $e$ of color $c$ on $\pi(a_i, t)$.
Then
\begin{equation}\label{eq:distances}
    d_{i+1} < \dist_G (e, t) \leq \dist_G (a_i, t) \leq \dist_{G-c} (a_i, t) \leq d_i
\end{equation}
where the first inequality holds as case~\ref{case2} did not occur, 
the second as $e$ belongs to $\pi(a_{i+1}, t)$, 
and the last by Invariant~\ref{inv:a_i}.
Now, by Invariant~\ref{inv:E_i}, $E_i$ only contains edges whose distance to $t$ in $G$ larger than $d_i$, so~\Cref{eq:distances} implies that $e \notin E_i$.
We therefore define $E_{i+1} = E_i \cup \{e\}$, so $E_{i+1}$ contains $i+1$ $c$-colored edges of distance $> d_{i+1}$ to $t$ in $G$, as required to satisfy Invariant~\ref{inv:E_i} for $i+1$.
To find $a_{i+1}$, we note that
\begin{equation}\label{eq:more-distances}
    \dist_{G-c} (s,t) \geq \dist_{G-c} (a_i, t) > d_{i+1}
\end{equation}
where the first inequality is by Invariant~\ref{inv:a_i}, and the second is by \Cref{eq:distances}.
Hence, by Property~\ref{prop:hitting-set}, there exist some $a_{i+1} \in A_{i+1}$ lying in the $d_{i+1}$-length suffix of $\pi(s,t \mid c)$, as required to satisfy Invariant~\ref{inv:a_i} for $i+1$.

We now finalize the correctness argument.
We have shown that the process can halt at some step $i < \DELTA$ only in case $\LastE(s,t \mid c) \in H$, so it remains to consider the case where $a_\DELTA$ and $E_{\DELTA}$ were constructed.
By Invariant~\ref{inv:E_i}, $E_\DELTA$ must contain \emph{all} of the $c$-colored edges in $G$ (as $|E_\DELTA| = \DELTA$, and color classes all have size $\leq \DELTA$), meaning they all have distance $>d_\DELTA$ from $t$ in $G$.
Also, by Invariant~\ref{inv:a_i}, $a_\DELTA \in A_{\DELTA}$ lies on $\pi(s,t \mid c)$ and $\dist_G (a_\DELTA, t) \leq \dist_{G-c} (a_\DELTA, t) \leq d_\DELTA$.
Thus, $\pi(a_{\DELTA},t)$ cannot contain any $c$-colored edges, so
$
\LastE(s,t \mid c) = \LastE(a_{\DELTA}, t \mid c) = \LastE(a_{\DELTA}, t)
$
and the latter edge belongs to $H$ by Rule~\ref{rule1}.

\section{Neighborhood Cover}\label{sec:neighborhood-cover}

Here, we give the proof of~\Cref{thm:nbr-cover-with-sep} from the classical $\tilde{O}(m)$ time algorithm for neighborhood covers by Awerbuch, Berger, Cowen and Peleg~\cite{ABCP98}.
First, by using their algorithm to construct a ``standard'' $2r$-neighborhood cover, we can compute a collection of clusters $Z_{ij} \subseteq V$ for $1 \leq i \leq p$, $1 \leq j \leq q$, and a mapping $\cover : V \to \{Z_{ij}\}_{1\leq i\leq p, 1 \leq j \leq q}$, such that:
\begin{itemize}[itemsep=0pt]
    \item (Covering) For each $v \in V$, $B_G (v, 2r) \subseteq \cover (v)$.
    \item  (Diameter) Each $Z_{ij}$ has (strong) diameter $O(r \log n)$.
    \item (Overlap) $p = O(\log n)$ and, for every $i=1,\dots p$, the sets $\{Z_{ij}\}_{1\leq j\leq q}$ are mutually disjoint.
\end{itemize}

We construct our ``shelled'' neighborhood cover for balls of radius $r$ as follows:
For each $Z_{ij}$ we create the kernel $K_{ij} = B_G (\cover^{-1}(Z_{ij}), r)$ and cluster $X_{ij} = B_G (K_{ij}, r) = B_G (\cover^{-1}(Z_{ij}), 2r)$ (and the shell $Y_{ij} = X_{ij} - K_{ij}$).
Let us now prove the required properties:
\begin{itemize}[itemsep=0pt]
    \item (Covering) For $v \in V$, if $\cover(v) = Z_{ij}$, then $B_G (v,r) \subseteq K_{ij}$ 
    by definition of $K_{ij}$.

    \item (Diameter) If $u,v \in X_{ij}$, then by definition of $X_{ij}$ there are $u',v' \in \cover^{-1}(Z_{ij}) \subseteq Z_{ij}$ such that $\dist_G (u,u'), \dist_G (v,v') \leq 2r$.
    By triangle inequality, $\dist_G (u,v) \leq 4r + \dist_G (u',v')$.
    But $u',v' \in Z_{ij}$ which has diameter $O(r \log n)$, so we obtain $\dist_G (u,v) = O(r \log n)$.
    Thus, $X_{ij}$ has weak diameter $O(r \log n)$.
    
    \item (Overlap) We assert that for each $1 \leq i \leq p$, the clusters $\{X_{ij}\}_{1 \leq j \leq q}$ are mutually disjoint, so in particular each vertex can belong to at most $p = O(\log n)$ clusters.
    Seeking contradiction, suppose $v \in X_{ij} \cap X_{i\ell}$ for $j \neq \ell$.
    Then there exist $u_j \in \cover^{-1} (Z_{ij})$ and $u_\ell \in \cover^{-1}(Z_{i\ell})$ such that $v \in B_G (u_j, 2r) \cap B_G (u_\ell, 2r) \subseteq Z_{ij} \cap Z_{i\ell}$, which is impossible as $Z_{ij}$ and $Z_{i\ell}$ are disjoint.
\end{itemize}

As for the running time: Computing $K_{ij}$ and $X_{ij}$ is done by multi-origin BFS%
\footnote{By multi-origin BFS we mean adding a dummy source vertex, connecting it by dummy edges to the set of origins, and running BFS from the dummy source.}
from $\cover^{-1}(Z_{ij})$ trimmed at depth $2r$, so the running time is linear in the number of edges touching $X_{ij}$.
For a fixed $1 \leq i \leq p$, the clusters $\{X_{ij}\}_j$ are disjoint, so the time to compute them is $O(m)$.
Hence, the overall running time is $\tilde{O}(m) + O(pm) = \tilde{O}(m)$.
This concludes the proof of~\Cref{thm:nbr-cover-with-sep}.

\end{document}